\documentclass{article}
\usepackage{amsmath}
\usepackage{amssymb}
\usepackage{amsthm}
\usepackage{hyperref}
\usepackage[capitalise]{cleveref}
\usepackage[sort,compress,numbers]{natbib}
\usepackage{xcolor}
\usepackage{graphicx}
\usepackage[margin=1.15in]{geometry}
\usepackage{caption}
\usepackage{subcaption}
\usepackage{wrapfig}

\newtheorem{theorem}{Theorem}[section]
\newtheorem{lemma}{Lemma}[section]

\newcommand{\dwt}{\widetilde{\Delta}}
\newcommand{\E}{\mathbb{E}}

\title{Multiserver-job Response Time under Multilevel Scaling}
\author{Isaac Grosof\footnote{\url{izzy.grosof@northwestern.edu}, Northwestern University, Department of Industrial Engineering and Management Sciences.}~
and Hayriye Ayhan\footnote{\url{hayriye.ayhan@isye.gatech.edu}, Georgia institute of Technology, H. Milton School of Industrial and Systems Engineering.}}

\begin{document}

\maketitle

\begin{abstract}
    We study the multiserver-job setting in the load-focused multilevel scaling limit, where system load approaches capacity much faster than the growth of the number of servers $n$.

    We consider the ``1 and $n$'' system, where each job requires either one server or all $n$.
    Within the multilevel scaling limit, we examine three regimes: load dominated by $n$-server jobs, 1-server jobs, or balanced.
    In each regime, we characterize the asymptotic growth rate of the boundary of the stability region and the scaled mean queue length.

    \textcolor{black}{We demonstrate that mean queue length peaks near balanced load via theory, numerics, and simulation.}
\end{abstract}

\section{Introduction}

Queueing theory primarily emphasizes models where only a single job can be served at a time, or where multiple jobs can be served but each requires the same amount of resources.
Both classes of models, single-server and homogeneous-multi-server models, capture some fraction of the space of real queueing systems while remaining conducive to theoretical study. However, homogeneous models do not reflect the behavior of many important systems, including many modern computing systems.

\begin{figure}
    \centering
    \includegraphics[width=0.8\linewidth]{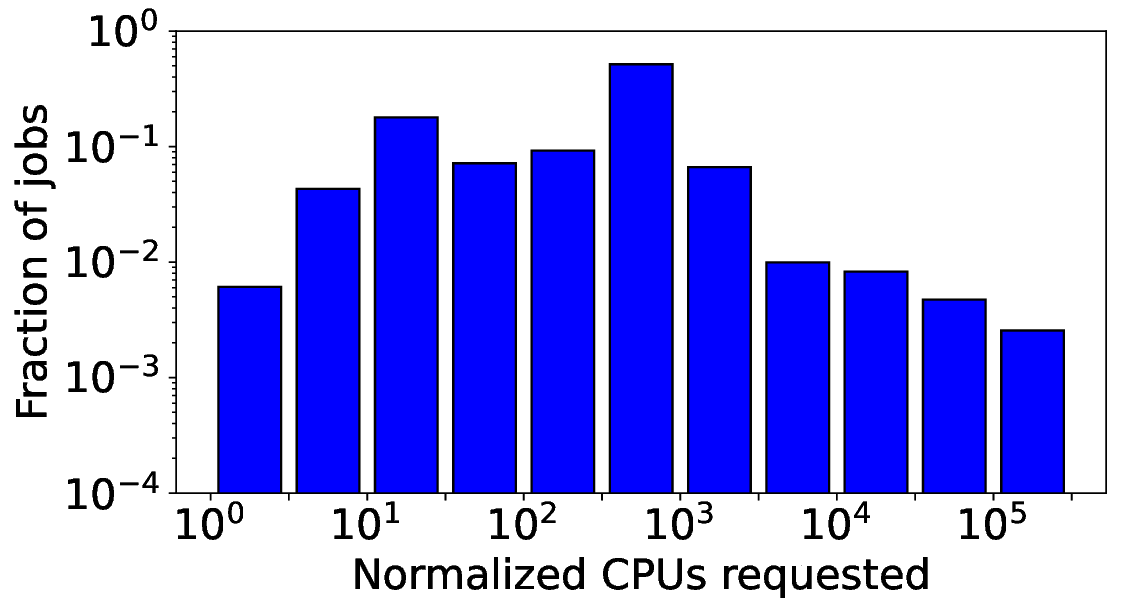}
    \caption{The distribution of number of CPUs requested
    in Google's Borg trace \cite{tirmazi_borg_2020}.
    Number of CPUs is normalized to the size of the smallest request observed,
    not an absolute value.
    The peak of the distribution is around 500 normalized CPUs,
    and there is significant probability mass anywhere from $1$ to $10^5$ normalized CPUs.
}
    \label{fig:cpu-request}
\end{figure}

For example, in modern datacenters, such as those of Microsoft, Google, Amazon, etc., each job requests an amount of computing resources (cores, processors, etc.) specific to the needs of that job. In Google's recently published trace of its ``Borg'' computing cluster \cite{tirmazi_borg_2020}, the requested resources of different jobs vary by a factor of $10^5$, as shown in \cref{fig:cpu-request}.
Models in which different jobs request different amounts of resources are also important for bandwidth sharing applications, where a job's resource requirement is its bandwidth need, rather than computing resources, or in high-performance computing settings with dedicated computing nodes.

Throughout this paper, we therefore focus on a ``multiserver-job model'' (MSJ), in which each job requests some number of servers, and concurrently occupies that many servers throughout its time in service.

There has been recent progress in analyzing the performance of MSJ queueing models, both under basic First-Come First-Served (FCFS) service \cite{grosof_reset_marc_2023}, as well as under more advanced scheduling policies \cite{grosof_wcfs_2021,anselmi_balanced_2025,grosof_optimal_2022}.
The scheduling policies which are popular in today's datacenters focus on heuristics closely based off of FCFS scheduling, such as backfilling heuristics \cite{srinivasan_characterization_2002,carastan_one_2019,wang_application_2009}.
Analyzing FCFS scheduling is therefore important to understanding the behavior of these systems.

Recent results on mean response time in MSJ FCFS models have focused on asymptotic limits, both the conventional heavy-traffic limit \cite{grosof_reset_marc_2023},
and a dual-scaling limit, where both the load and the number of servers grow asymptotically \cite{wang_zero_2021,hong_sharp_2024}.
The latter work focuses on scalings where the number of servers grows so quickly relative to the load that the system experiences zero queueing or zero expected waiting time in the asymptotic limit.
We call this limit the ``server-focused multilevel scaling'' limit.

This recent work leaves open the problem of performance analysis in the ``load-focused multilevel scaling'' (LFMS) limit, where load approaches capacity much faster than the asymptotic growth of the number of servers.
Stability results and numerical results exist in the LFMS limit, but no asymptotic scaling results exist.
In this paper, we analyze the asymptotic scaling of mean queue length of a MSJ system in the LFMS limit. To the best of our knowledge, this is the first such analysis. In particular, we inaugurate the study of the LFMS limit by focusing on a specific MSJ system. We study the ``1 and $n$'' system, where each job's server need is either a single server, or all $n$ servers in the system, with FCFS service. \textcolor{black}{Our main objective in considering a system with two extreme server needs is to understand the impact of this variability in server needs on system performance. The insights derived from our setting are applicable in more general settings with highly-variable server needs across more than two classes. As a result, our insights have practical implications since systems with extreme server needs arise in applications such as Large Language Model training \cite{wan_robust_2025} and \cite{Mega} and supercomputers such as Frontier \cite{Frontier}}.

\textcolor{black}{Our key insight, as discussed in \cref{extensions}, is that settings with load spread evenly across high and low-resource job types exhibit much higher mean queue lengths than settings with load concentrated on just high-resource jobs, or just low-resource jobs, controlling for system utilization. This insight generalizes beyond the 1-and-$n$ system that we study, to cover general high-variability MSJ settings, such as AI training and supercomputing.}

Within the LFMS limit in the 1 and $n$ system, we encounter three regimes with distinct behavior, depending on the fraction of system load corresponding to jobs of each type.
We call these regimes the 1-server dominated regime, the balanced regime, and the $n$-server dominated regime.

We provide the first analysis of mean queue length scaling with respect to the number of servers $n$ in each of these three LFMS regimes.

Specifically, we characterize the asymptotic growth rate of the completion rate $\mu$ and the scaled mean queue length $\E[Q(1-\rho)]$, where $Q$ is the queue length and $\rho = \lambda/\mu$ is the fraction of capacity in use.
We characterize these growth rates in each of three LFMS regimes: The $n$-server dominated, balanced, and 1-server dominated regimes, \textcolor{black}{in \cref{thm:n-server-dominated,thm:1-server-dominated} respectively. Note that as demonstrated in \cref{thm:n-server-dominated}, the asymptotics are the same in the first two regimes.}
We also empirically evaluate our asymptotic expressions in comparison to prior fixed-$n$ characterizations in \cref{sec:empirical}.

The rest of the paper  is organized as follows:
\cref{sec:prior} reviews relevant literature in relation to our results.
\cref{sec:model} provides the details of our MSJ model and the related saturated system.
\cref{sec:main} states our main results.
\cref{sec:prelim} provides preliminary characterizations of system performance.
\cref{sec:regime-1,sec:regime-2,sec:regime-3} prove our main results in the three regimes discussed above. \textcolor{black}{\cref{sec:empirical} compares our asymptotic results against fixed-$n$ formulas. Section \ref{extensions} focuses on generalizations of our original $1$ and $n$ system. Finally, Section \ref{conclusions} concludes the paper.}


















\section{Prior work}
\label{sec:prior}

We review prior theoretical work on multiserver-job (MSJ) systems -- for an initial overview, see \cite{harchol_multiserver_2022}.
We discuss the MSJ stability region under the First-Come First-Served (FCFS) policy in \cref{sec:fcfs-stability}, mean response time under FCFS in \cref{sec:fcfs-response}, and results for other scheduling policies in \cref{sec:other-scheduling}.
\subsection{Stability of the MSJ system under FCFS}
\label{sec:fcfs-stability}

In the MSJ setting, the First-Come First-Served (FCFS) policy with head-of-line blocking does not consistently keep all servers occupied, reducing its stability region relative to more optimized policies.
Characterizing this stability region is an important line of work.

An important approach for characterizing the stability region of an MSJ system under FCFS is the \emph{saturated system} approach \cite{baccelli_saturation_1995}. This work provides a general framework under which one can relate the stability region of a queueing system which completes jobs in near-FCFS order to the throughput of a \emph{saturated} queueing system.
In the MSJ system, jobs enter service in FCFS order, causing completion order to be very close to FCFS order.
In the saturated system, an unlimited amount of jobs are available at all times, replacing the external arrival process.

This approach  has been used to characterize the stability region of a variety of MSJ FCFS systems \cite{grosof_reset_marc_2023}, though often requiring the solution of a system of linear equations to obtain the stability region.
More explicit symbolic results have been proven for special cases of job duration distribution and server need \cite{grosof_new_2023,rumyantsev_stability_2017,rumyantsev_stability_2020}.

These results have also been generalized to certain asymptotic limits, in the case of exponential duration and 1-or-$n$ server need \cite{olliaro_saturated_2023},
specifically the limit as the $n$-server jobs complete much faster than the $1$-server jobs, and the limit as the number of servers $n$ diverges.

We study the same 1-and-$n$ exponential setting as \cite{olliaro_saturated_2023}.
The asymptotic limits in that paper are simpler than our asymptotic limit -- we allow both $n \to \infty$ and $p_n \to 0$\textcolor{black}{, where $p_n$ is the probability that a job needs all $n$ servers,} allowing us to explore a richer variety of asymptotic behavior. Our results cover both stability region and mean response time.

\subsection{Response time of MSJ system under FCFS}
\label{sec:fcfs-response}

Much less is known about multiserver-job (MSJ) FCFS mean response time.
The only setting in which mean response time has been exactly characterized is the setting with $n=2$ servers and where 1-server jobs require an exponential service duration \cite{brill_queues_1984,fillippopoulos_mm2_2007}.

In a related but distinct model has been considered where a job's service must start at the same time at each of its servers, but may end at different times on different servers.
If completion times are independent across the servers, exact results on mean waiting time and the distribution of waiting time are known \cite{green_queueing_1980,seila_technical_1984}. In contrast, the MSJ model requires that service ends at the same time across all of a job's servers.

In more general MSJ settings, one must turn either to asymptotic results or numerical algorithms and approximations.

As for numerical results for response time, the standard technique is to employ matrix geometric and matrix analytic methods, and we see no reason why off-the-shelf methods would not be applicable \cite{latouche_introduction_1999}. These methods were also used to derive some of the stability results discussed above \cite{rumyantsev_stability_2017}. Such numerical results are far more efficient than simulation, but are inherently restricted to examining a single set of parameters at a time. They cannot provide asymptotic results, which are the focus of this paper.

Mean response time, or equivalently mean waiting time, has been explored in several asymptotic limits.
One type of limit that has received significant attention is a scaling limit where the number of servers, the largest server need in the system, and the load all scale polynomially, mimicking the classical Halfin-Whitt limit \cite{wang_zero_2021,hong_sharp_2024}.
However, results in this limit have been limited to the case where the load is low enough to allow stability with the maximal number of wasted servers: If the largest server need of a job is $m_{\max}$, this line of work requires that the mean number of servers occupied is below $n - m_{\max}$.
This limit becomes trivial when $m_{\max} = n$, as in this work.
Results in this limit focus on the asymptotic decay of queueing probability \cite{wang_zero_2021} or of mean waiting time in the scaling limit \cite{hong_sharp_2024}.

We refer to this limit as the server-focused multilevel scaling (SFMS) limit.
The SFMS limit has much lighter load than the \textcolor{black}{LFMS} limit of this paper.

The remaining asymptotic work has focused on the fixed-$n$ limit, as load approaches capacity. This limit mirrors the classical heavy-traffic limit.
Tight asymptotic characterization of mean response time is known for MSJ FCFS in this limit \cite{grosof_reset_marc_2023}. However, these results do not immediately generalize to any multilevel-scaling limit. This generalization is the contribution of this paper.

\subsection{MSJ system under other scheduling policies}
\label{sec:other-scheduling}

Outside of FCFS scheduling, more results are known in the MSJ setting under nontrivial scheduling policies.

Throughput-optimality results have been established both in a preemptive setting, for the MaxWeight policy \cite{maguluri_stochastic_2012},
and in a nonpreemptive setting,
for the Randomized Timers policy \cite{psychas_randomized_2018}, though the Randomized Timers policy has extremely poor response time in certain settings.
The recent Markovian Service Rate policy uses a simple timer-based scheme to achieve throughput optimality, including in a nonpreemptive setting, with predictable and acceptable response time in all settings \cite{chen_analyzing_2024}.

Moving to policies whose analyses primarily focus on mean response time, a significant line of work has focused on settings where all jobs' server needs are powers of 2 and $n$ is a power of 2, and settings with related divisibility assumptions. First, the preemptive ServerFilling policy has been shown to match the response time of a resource-pooled M/G/1/FCFS, in the power-of-2 setting with a fixed number of servers $n$ \cite{grosof_wcfs_2021,grosof_serverfilling_2023}.
In this setting, the ServerFilling-SRPT policy has been shown to achieve optimal mean response time in heavy traffic \cite{grosof_optimal_2022}.
Finally, in the power-of-2 setting, the Balanced Splitting dispatching policy has been exactly analyzed, with both fixed-$n$ and scaling analysis results \cite{anselmi_balanced_2025}.

In the SFMS limit discussed in \cref{sec:fcfs-response},
the Smallest Need First policy has also been studied, resulting in an expanded set of asymptotic limits in which expected waiting time decays to 0 \cite{hong_sharp_2024}.

Many practically used scheduling policies are backfilling policies, which can be seen as starting with the FCFS scheduling policy, with additional steps taken to use capacity that FCFS would otherwise leave unused \cite{srinivasan_characterization_2002,carastan_one_2019,wang_application_2009}. These policies have been investigated in simulation and in practice, but little is known about their theoretical properties. We hope to strengthen the theoretical understanding of the FCFS policy, to build towards a theoretical understanding of backfilling policies.

\section{Model}
\label{sec:model}

We introduce the multiserver-job (MSJ) model in \cref{sec:general-msj},
and discuss
the saturated system, a key theoretical tool for analyzing the MSJ system, in \cref{sec:def-saturated}.

\subsection{Multiserver-job system}
\label{sec:general-msj}

We consider a multiserver-job (MSJ) system with  $n$ identical servers. Customers arrive according to a Poisson process with rate $\lambda$. We specifically consider a model in which there are two types of jobs: Jobs which need to be served by one server, and jobs which need all $n$ servers. The probability that a job needs one server is $p_1$, and \textcolor{black}{(as defined in Section \ref{sec:fcfs-stability})} the probability that a job needs all $n$ servers is $p_n = 1-p_1$, i.i.d. The service time for a 1-server job is $Exp(\mu_1)$, and for an $n$-server job is $Exp(\mu_n)$, both independent of all other jobs.
Jobs are served in First-Come First-Served (FCFS) order with head-of-line blocking, meaning that jobs are admitted into service in FCFS order until a job is reached which cannot be served in the remaining available servers. No further jobs are served.
We define $q$ to be the queue length at a particular point in time, and $Q$ to be the stationary queue length random variable.

Note that the symbol $p_n$ always refers to the fraction of jobs that use all of the servers,
even as the number of servers $n$ changes.

Prior work \textcolor{black}{(see for example \cite{grosof_new_2023,rumyantsev_stability_2017,rumyantsev_stability_2020}, \cite{grosof_reset_marc_2023}, or \cite{olliaro_saturated_2023})}  has demonstrated the existence of a stability threshold $\mu$, where the system is stable if and only if $\lambda < \mu$. Note that $\mu$ has been explicitly characterized by several methods, including analysis of the corresponding saturated system, which we define and discuss further in \cref{sec:def-saturated}. We build on these prior explicit results to establish the scaling of the stability threshold $\mu$. We define $\rho = \lambda / \mu$ to be the fraction of the stability region that is in use.

A natural asymptotic limit is the $\lambda \to \mu$ scaling limit, known as the ``heavy traffic'' limit.
The limiting value of the mean scaled queue length, $\lim_{\rho \to 1} \E[Q(1-\rho)]$, has been characterized by analyzing the saturated system, which we likewise discuss in \cref{sec:def-saturated}.

However, the focus of this paper is the ``load-focused multilevel scaling'' (LFMS) limit where we allow $\rho$ to grow quickly enough that $\E[Q(1-\rho)]$ reaches its limiting value, and only then scale the number of servers to study the secondary effects of that scaling.

When scaling the number of  servers, it is natural to also consider scaling the fraction of single-server jobs $p_1$ towards 1, as the number of single-server jobs that can be served at once grows asymptotically.

Our goal is to characterize asymptotic behavior of the mean scaled queue length for this system under two secondary asymptotic limits:  $p_1 \to 1$, and $n \to \infty$. Both of these asymptotic limits are secondary to our load scaling.

In particular, it makes sense to consider taking both limits simultaneously as there is  a natural separation, based on whether the system load is dominated by 1-server jobs, by $n$-server jobs, or is balanced between the two.

We refer to the inherent workload of a job as the product of its server need and its duration. We refer to the overall load as the rate at which this workload is arriving to the system. The overall load of 1-server jobs is $\lambda \frac{p_1}{\mu_1}$, and the overall load of $n$-server jobs is  $\lambda \frac{n p_n}{\mu_n}$.
As a result, 1-server jobs require $\frac{p_1}{p_1 + n p_n}$ fraction of the total load.

We focus on the limit where $n \to \infty$ and $p_n \to 0$, holding $\mu_1$ and $\mu_n$ constant, and adjusting $\lambda$ as necessary to focus on the $\rho \to 1$ limit.
As $p_1 \to 1$, the fraction of the total load required by 1 server jobs can be approximated as $\frac{1}{1 + n p_n}$.

In this paper, we define the notation $f(n) = o(g(n))$ and $f(n) = O(g(n))$ as follows:
\begin{align*}
    f(n) = o(g(n)) \Leftrightarrow \lim_{n \to \infty} \frac{f(n)}{g(n)} = 0, \quad
    f(n) = O(g(n)) \Leftrightarrow \lim_{n \to \infty} \frac{f(n)}{g(n)} \in (-\infty, \infty).
\end{align*}
Note that when $f(n) = o(g(n))$ or $f(n) = O(g(n))$, $f(n)$ may be negative while $g(n)$ is positive.

We define $\omega(\cdot), \Omega(\cdot),$ and $\theta(\cdot)$ similarly.
If $f(n) = o(g(n)),$ then $g(n) = \omega(f(n))$.
If $f(n) = O(g(n)),$ then $g(n) = \Omega(f(n))$.
If $f(n) = O(g(n))$ and $f(n) = \Omega(g(n))$, then $f(n) = \theta(g(n))$.

If $n p_n = \omega(1)$, or equivalently $p_n = \omega(\frac{1}{n})$, the fraction of load occupied by 1-server jobs converges to 0 as $p_n \to 0, n \to \infty$.
We call this the ``$n$-server dominated'' regime.

If \textcolor{black}{$p_n = c/n$, for an arbitrary positive constant $c$  (which would imply that $n p_n = \theta(1)$, or equivalently $p_n = \theta(\frac{1}{n})$)}, the fraction of load occupied by 1-server jobs converges to a value strictly between 0 and 1. We call this the ``balanced'' regime.

Finally, if $n p_n = o(1)$, or equivalently $p_n = o(\frac{1}{n})$, the fraction of load occupied by 1-server jobs converges to 1.
We call this the ``1-server dominated'' regime.
We specifically focus on the polynomial-scaling subset of the 1-server dominated regime, where $p_n = 1/n^\alpha, \alpha > 1$.

\textcolor{black}{Note that through out the paper, all limits are taken as the number of servers $n$ tends to infinity.}

\subsection{Saturated System}
\label{sec:def-saturated}

A key tool for analyzing the MSJ system is the ``saturated system''.
The saturated system is a closed system with the same MSJ completion setup as the original open system.
In the saturated system, whenever a server is unoccupied and there is no job blocking the head of the line,
jobs enter the system with the same i.i.d. probabilities $p_1$ and $p_n$ of 1 and $n$-server jobs, respectively.
There is no separate stochastic arrival process. Jobs enter in response to a completion until either all $n$ servers are occupied, or an $n$-server job is generated that cannot  be served immediately.

Saturated system analysis can be used to derive the stability region of the original MSJ system, because the stability region of the original system is bounded by the throughput of the saturated system \cite{grosof_reset_marc_2023,baccelli_saturation_1995,foss_overview_2004}.

Saturated system analysis can also be used to derive the mean queue length of open MSJ systems under FCFS service.
To do so, \cite{grosof_reset_marc_2023} makes use of the \emph{relative completions} function $\Delta$, which maps states of the saturated system to real values.
Note that the concept of relative arrivals and relative completions was also present in work prior to that result \cite{falin_heavy_1999,dimitrov_single_2011}.

To define $\Delta$, we must first define several quantities in the saturated system.
Let $y$ denote a state of the saturated system and define $\mu_y$ as the completion rate in state $y$. Let $\nu_y$ denote the total transition rate out of state $y$, and $\nu_{y,y'}$ denote the specific transition rate from state $y$ to another state $y'$, where $y'$ is an arbitrary state of the saturated system. Finally, define $\mu$ as the time-average completion rate of the saturated system (e.g. the throughput).

Letting $y$ denote a state of the saturated system,
$\Delta(y)$ is defined as the solution to the following system of equations:
\begin{align}
    \label{eq:poisson}
    \Delta(y) := \frac{\mu_y - \mu}{\nu_y} + \sum_{y'} \frac{\nu_{y,y'}}{\nu_y} \Delta(y').
\end{align}
This equation \eqref{eq:poisson} can be seen as the Poisson equation for a Markov Reward Process whose instantaneous reward is the completion rate $\mu_y$. \textcolor{black}{This relative completions function impacts mean queue length in the MSJ system similarly to the impact of job size variance in an M/G/1 queue: It effectively measures the variability in completion rate. If the system is in a state $y$ with a large value of $\Delta(y)$, the system is likely to experience a high completion rate in the near future.}

Note that the solution to this system of equations is only defined up to an additive constant. Letting $Y$ denote the time-average random variable for the saturated system state, \cite{grosof_reset_marc_2023} adopts the convention that $\E[\Delta(Y)] = 0$,
which results in a unique solution.

Let $Y_d$ denote the state-average random variable corresponding to the embedded discrete-time Markov chain (DTMC), updating at completion epochs.
\cite[Theorem 4.2]{grosof_reset_marc_2023} proves the following result, for an arbitrary MSJ FCFS system, in the heavy traffic limit:
\begin{align}
    \label{eq:delta-limiting-scaled}
    \lim_{\rho \to 1} \E[Q(1-\rho)] = \E[\Delta(Y_d)] + 1.
\end{align}

Our goal in this paper is to analyze the asymptotic growth rate of $\E[\Delta(Y_d)]$ in each of our three regimes, within the surrounding load-focused multilevel scaling context.


\section{Main results} \label{sec:main}


We study the heavy traffic limit under load-focused multilevel scaling in each of three regimes:

\begin{enumerate}
    \item $n$-server dominated: $p_n = \omega(1/n)$,
    \item Balanced: $p_n = c/n$, where $c > 0$ is an arbitrary positive constant,
    \item 1-server dominated: $p_n = o(1/n)$, with a specific focus on the case where $p_n = 1/n^\alpha, \alpha > 1$.
\end{enumerate}

The main results of the paper are characterizations of $\E[\Delta(Y_d)]$, which determines the scaled mean queue length $\E[Q(1-\rho)]$ as discussed in \cref{sec:def-saturated}, and $\mu$, which determines the stability region, in the limit as $p_n \to 0$ and $n \to \infty$, in each of these three regimes.

We now state our results for each of these regimes:

\begin{theorem}
    \label{thm:n-server-dominated}
\textcolor{black}{ In regime 1 which is  the $n$-server dominated regime, namely $p_n = \omega(1/n)$, as $p_n \to 0, n \to \infty$, and in regime 2,  which is the balanced regime, namely $p_n = c/n$, for some positive constant $c$, as $p_n \to 0, n \to \infty$},
    \begin{align*}
         \mu &=  (1+\textcolor{black}{o(1)}) \frac{\mu_1}{p_n \ln(1/p_n)} \mbox{ and}\\
        \E[\Delta(Y_d)] &= (1+\textcolor{black}{o(1)})\frac{1}{2p_n}.
    \end{align*}
\end{theorem}
\begin{proof}{\textcolor{black}{The proof for  regime 1 is given in \cref{sec:regime-1}, split into \cref{lem:r1-mu} and \cref{thm:r1-delta} and 
the proof for regime 2 is given in \cref{sec:regime-2}, split into \cref{lem:r2-mu} and \cref{thm:r2-delta}.}} 
\end{proof}

\textcolor{black}{Intuitively, in the $n$-server dominated regime, our result shows that the throughput is dominated by 1-server jobs which are served while there is an $n$-server job at the head of the queue, and most of the servers are idle.
This is why $\mu$ is far below $n \mu_1$ (the service rate when all servers are busy with 1-server jobs).
This gives rise to periodic fluctuation in service rate with cycles of length $\frac{1}{p_n}$ (the number of $1$-server jobs between each $n$-server job). These fluctuations control the size of $\E[\Delta(Y_d)]$,
and hence the scaled mean queue length.}


Note that the asymptotic growth rates in the balanced regime do not depend on $c$. \textcolor{black}{The intuition behind this is that  the balanced load regime exhibits the same behavior as regime 1, namely,  throughput and scaled mean queue length are dominated by 1-server jobs served while there is an $n$-server job at the head of the queue, and most servers are idle.}
\textcolor{black}{Regime 2 represents the boundary of parameters exhibiting this behavior,} whereas in regime 3 we show that a distinct behavior emerges.

\textcolor{black}{Phrased another way, regime 2 represents the peak of mean queue length, when controlling for utilization: $\E[\Delta(Y_d)]$ rises until regime 2 is reached, and falls thereafter in regime 3. See \cref{extensions} for an exploration of this mean-queue-length peak via numerical evaluation and simulation, and its generalization into further settings.}

In regime 3, the 1-server dominated regime, we primarily focus on the setting where $p_n$ is a polynomial function of $n$, namely $p_n = 1/n^\alpha$ for some constant $\alpha > 1$.
This setting can be seen as equivalent to the generalized Halfin-Whitt regime, which has been explored for instance
in a load-balancing setting \cite{liu_steady_2020,liu_universal_2022}.

\begin{theorem}
    \label{thm:1-server-dominated}
    In regime 3, which is the 1-server dominated regime, namely  $p_n = 1/n^\alpha$ for some constant $\alpha > 1$,
    as $p_n \to 0, n \to \infty$,
    \begin{align*}
        \mu &= \mu_1 \left(n - n^{2 - \alpha} (\ln n + \mu_1/\mu_n - 1 +\textcolor{black}{o(1)}) \right) \mbox{and} \\
        \E[\Delta(Y_d)] &= (1+\textcolor{black}{o(1)}) \frac{1}{2} n^{2 - \alpha} \ln^2 n.
    \end{align*}
\end{theorem}
\begin{proof}[Proof is given in \cref{sec:r3-spec}, split into \cref{lem:r3-spec-mu} and \cref{thm:r3-spec-delta}]
\end{proof}

In the 1-server dominated regime, we see a different behavior than in the two previous regimes. Here, throughput nearly matches $n \mu_1$, the service rate when all $n$ servers are serving $1$-server jobs. The service that occurs when all $n$ servers are serving $1$-server jobs dominates the behavior of the system.
Correspondingly, we see smaller values of $\E[\Delta(Y_d)]$.
Specifically, in this regime, the expected number of small jobs between a pair of large jobs is still $1/p_n$, but  that period is spent predominantly with a consistent $n \mu_1$ service rate, causing $\E[\Delta(Y_d)]$ to fall from the peak of $\frac{1}{2} \frac{n}{c}$ reached in the balanced load regime, with slower and slower growth rates for larger values of $\alpha>1$.

We also prove more general but less detailed results in the general 1-server dominated regime, in \cref{sec:r3-gen} (see \cref{lem:r3-gen-mu} and \cref{lem:r3-gen-delta}).

\cref{sec:prelim} provides some preliminary analysis and introduces the notation that will be used throughout the paper.
\section{Preliminary Analysis -- Saturated System}
\label{sec:prelim}

In this section, we study the saturated MSJ system.
There are three stationary distributions of importance in this system:
the time-average stationary distribution $P$,
and two distributions based on embedded Markov chains.
Note that the random variable $Y$ is distributed according to $P$.
The first embedded DTMC that we consider is the transition-based DTMC which updates on either arrivals or completions, giving rise to the transition-average stationary distribution $\pi$.
Note that we separate each arrival as its own transition in this embedded DTMC, even if multiple arrivals occur at the same time.
The second embedded DTMC is the completion-based DTMC which only updates on completion,
giving rise to the completion-average stationary distribution $\pi^{d}$.
Note that the random variable $Y_d$ is distributed according to $\pi^{d}$.

 In this section, we will characterize $\pi$, $P$, $\pi^d$, and the throughput $\mu$, for a given $n$.
 We will use this characterization to derive the asymptotic growth rates of $\mu$ and $\E[\Delta(Y_d)]$.

\subsection{Transition-average stationary distribution}
\label{sec:prelim-transition}

The state of the saturated system in the transition-based embedded DTMC
can be captured by a pair $(a, b)$
where $a$ denotes the number of $n$-server jobs present, $a \in \{0, 1\}$,
and  $b$ denotes the number of $1$-server jobs present, $0 \le b \le n$.

There are two kinds of states in this embedded DTMC:
states where jobs are completing, and states where jobs are arriving.
If $a = 0$ and $b < n$, then the state is an arriving state, because room for additional jobs is available,
while if $a = 1$ or $b = n$, the state is a completing state.
If $a = 1$ and $b = 0$, then an $n$-server job is in service, while otherwise only 1-server jobs are in service.

The transition probabilities are as follows:
\begin{itemize}
    \item From an arriving state  $(0, b)$, where $b < n$, the arriving job is a 1-server job with probability $p_1$, resulting in a transition to $(0, b+1)$, or an $n$-server job with probability $p_n$,
    resulting in a transition to $(1, b)$.
    \item From a completing state of the form $(1, b)$, where $1 \le b < n$,
    $b$ 1-server jobs are in service, resulting in a guaranteed transition to $(1, b-1)$.
    \item From the state $(1, 0)$, the $n$-server job is in service, resulting in a transition to $(0, 0)$.
    \item From the state $(0, n)$, $n$ 1-server jobs are in service, resulting in a transition to $(0, n-1)$.
\end{itemize}

Now, we characterize the transition-average stationary distribution $\pi$:

\begin{lemma} \label{pi}
    \label{lem:pi}
    The transition-average stationary distribution $\pi$ can be written as:
    \begin{align*}
    \forall b < n-1, \pi_{0, b} &= \pi_{0, 0} p_1^b, \\
    \pi_{0, n-1} &= \pi_{0, 0} \frac{p_1^{n-1}}{p_n}, \\
    \pi_{0, n} &= \pi_{0, 0} \frac{p_1^n}{p_n}, \\
    \forall b < n, \pi_{1, b} &= \pi_{0, 0} p_1^b.
\end{align*}   
\end{lemma}
\begin{proof}
Note that the arrivals occur in states $(0, b)$ for $0 \le b < n$, and job service completions occur in states  $(0, n)$ and $(1, b)$ for $0 \le b < n$. Furthermore, in  state $(1, 0)$, the $n$-server job is in service. In all other job service completion states, only $1$-server jobs are in service.  Then we have the following balance equations  
    \begin{align*}
    \pi_{0, 0} &= \pi_{1, 0}, \\
    \forall b,  0 < b < n-1, \pi_{0, b} &= \pi_{0, b-1} p_1, \\
    \pi_{0, n-1} &= \pi_{0, n-2} p_1 + \pi_{0, n}, \\
    \pi_{0, n} &= \pi_{0, n-1} p_1, \\
    \forall b, 0 \le b < n-1, \pi_{1, b} &= \pi_{1, b+1} + \pi_{0, b} p_n, \\
    \pi_{1, n-1} &= \pi_{0, n-1} p_n,
    \end{align*}
and the result follows.
\end{proof}

\subsection{Time-average stationary distribution}

We now use our results from \cref{sec:prelim-transition} to characterize the time-average stationary distribution.

\begin{lemma}
    \label{lem:p}
    The time-average stationary distribution $P$ is given as 
    \begin{align*}
        P_{1, 0} &= C' \frac{1}{\mu_n},\\
        \forall b, 1 < b < n, P_{1, b} &= \frac{1}{b \mu_1} = C' \frac{1}{b \mu_1} p_1^b, \\
        P_{0, n} &= C' \frac{1}{n \mu_1} \frac{p_1^n}{p_n},
    \end{align*}
where $\mu_1$ and $\mu_n$ is the service completion rate of size 1 and size  $n$ jobs, respectively and 
\begin{align*}
 C' &= \left( \frac{1}{\mu_n} + \frac{1}{n\mu_1} \frac{p_1^n}{p_n} + \sum_{b=1}^{n-1} \frac{p_1^b}{b \mu_1}\right)^{-1}.    
\end{align*}
\end{lemma}
\begin{proof}
    The time-average stationary distribution is proportional to the transition-average stationary distribution, rescaled by the time spent in each state. The only states in which time is spent are completion states, and the average time per visit is the inverse of the completion rate. Thus, we define a proportionality constant $C$, such that $P_y = C \frac{\pi_y}{\mu_y}$, where recall that $\mu_y$ is the completion rate in state $y$.
    We also define $C' = C \pi_{0, 0}$.
    Now, applying Lemma \ref{pi}, we have 
   \begin{align*}
        P_{1, 0} &= C\pi_{1, 0} \frac{1}{\mu_n} = C' \frac{1}{\mu_n},\\
        \forall b, 1 < b < n, P_{1, b} &= C\pi_{1, b} \frac{1}{b \mu_1} = C' \frac{1}{b \mu_1} p_1^b, \\
        P_{0, n} &= C\pi_{0, n}\frac{1}{n \mu_1} = C' \frac{1}{n \mu_1} \frac{p_1^n}{p_n},
    \end{align*} 
and $C'$ can be computed by using the fact that $P$ is a distribution:
    \begin{align*}
        1 &= \sum_y P_y = C' \left( \frac{1}{\mu_n} + \frac{1}{n\mu_1} \frac{p_1^n}{p_n} + \sum_{b=1}^{n-1} \frac{p_1^b}{b \mu_1}\right) \qedhere
    \end{align*}
\end{proof}

\begin{lemma}
    \label{lem:explicit-mu}
    The explicit expression for $\mu$ is given as 
    \begin{align*}
        \mu &= \frac{1}{p_n} \left( \frac{1}{\mu_n} + \frac{1}{n\mu_1} \frac{p_1^n}{p_n} + \sum_{b=1}^{n-1} \frac{p_1^b}{b \mu_1}\right)^{-1}.
    \end{align*}
\end{lemma}
\begin{proof}
We apply \cref{lem:p}. Note that we have
    \begin{align*}
        \mu &= P_{1, 0} \mu_n + P_{0, n} n \mu_1 + \sum_{b=1}^{n-1} P_{1, b} b \mu_1 \\
        &= C' (1 + \frac{p_1^n}{p_n} + \sum_{b=1}^{n-1} p_1^b) 
         = C' (1 + \frac{p_1^n}{p_n} + \frac{p_1-p_1^n}{1-p_1}) 
         = C' \frac{1}{p_n}.
    \end{align*}
Plugging in the closed form expression of $C'$ gives the desired result.
\end{proof}

    Note that semantically, $C'$ is the rate of completion of $n$-server jobs, and $\frac{1}{C'}$ is the expected time per $n$-server job completion. We have 
    \begin{align*}
        \frac{1}{C'} = \frac{1}{\mu_n} + \frac{1}{n\mu_1} \frac{p_1^n}{p_n} + \sum_{b=1}^{n-1} \frac{p_1^b}{b \mu_1}.
    \end{align*}

    Intuitively, one can see that there are three phases to the completion of an $n$-server job, namely, time spent in state $(1, 0)$, time spent in state $(0, n)$, and  time spent in the states $(1, b)$ for $0 < b < n$. The expected lengths of these phases are the first, second, and third  terms, respectively, in the above expression.

\subsection{Completion-average stationary distribution}

Now we characterize the completion-average distribution $\pi^d$.
Note that there are only $n+1$ possible states in which a completion can occur, with 0 to $n$ 1-server jobs in service. 

\begin{lemma}
    \label{lem:pi_d}
    The completion-average stationary distribution $\pi^d$ is given as
    \begin{align*}
        \forall b < n, \, \pi^d_{1, b} &= p_1^b p_n, \\
        \pi^d_{0, n} &= p_1^n.
    \end{align*}
\end{lemma}
\begin{proof}
    Note that the distribution $\pi^d$ is equal to the transition-average stationary distribution $\pi$, conditioned on the saturated system being in a completion state. In particular, there exists a normalizing constant $C''$ such that for any completing state $s$,
    \begin{align*}
        \pi^d_s = C'' \pi_s.
    \end{align*}
    Using the expression for $\pi$ given in \cref{lem:pi}, we can solve for the proportionality constant $C''$ as simply $\frac{p_n}{\pi_{0, 0}}$, giving rise to the stated expression for $\pi^d$.
\end{proof}

\subsection{Relative completions}

We now characterize the relative completions in the saturated system, which is key to the characterization of the scaling behavior of mean queue length in the original open system (see \cref{sec:def-saturated} for details).

\begin{lemma}
    \label{lem:explicit-delta}
    The explicit formula for $\E[\Delta(Y_d)]$ is
    \begin{align*}
        \E[\Delta(Y_d)] &=  p_1^n (1 - \frac{\mu}{n \mu_1}) ( n-1 + \frac{1}{p_n} - \frac{\mu}{\mu_1} H_{n-1} - \frac{\mu}{n p_n \mu_1}) \nonumber\\
   & +  \sum_{i=1}^{n-1} p_1^i p_n(1 - \frac{\mu}{i \mu_1}) (i - \frac{\mu}{\mu_1} H_i).
    \end{align*}
    \textcolor{black}{where $H_i=\sum_{k=1}^i\frac{1}{k}$ is the $i^{\mbox{th}}$ harmonic number.}
\end{lemma}

\begin{proof}
For notational convenience, we will now switch to a one-dimensional state representation.
More specifically, define the states of the time-average system and the completion-average embedded DTMC to be the integers
$\{0, 1, 2, \ldots, n\}$,
where the state is the number of 1-server jobs in the system.
Note that we ignore the arrival states $(0, 0), (0, 1), \ldots, (0, n-1)$,
as the system does not spend time in these states and completions cannot occur in them.

To characterize $\E[\Delta(Y_d)]$, we first characterize $\Delta$ up to an offset.
Let $\dwt(i) := \Delta(i) - \Delta(0)$, for any state $i$.
We will use the following identity:
\begin{align}
    \E[\dwt(Y)] = \E[\Delta(Y)] - \E[\Delta(0)] = - \E[\Delta(0)].
\end{align}

The second equality holds because $\E[\Delta(Y)]$ is defined to be 0, as discussed in \cref{sec:def-saturated}.
Thus, we can characterize $\E[\Delta(Y_d)]$ as follows:
\begin{align}
    \nonumber
    \E[\dwt(Y_d)] &= \E[\Delta(Y_d)] - \E[\Delta(0)] = \E[\Delta(Y_d)] + \E[\dwt(Y)] \\
    \label{eq:delta-dwt}
    \implies \E[\Delta(Y_d)] &= \E[\dwt(Y_d)] - \E[\dwt(Y)].
\end{align}

Thus, it suffices to characterize $\dwt$ and use \eqref{eq:delta-dwt} to characterize $\E[\Delta(Y_d)]$.

We can rewrite \eqref{eq:delta-dwt} more explicitly as 
\begin{align} \label{y_d}
    \E[\Delta(Y_d)] = \sum_{i=0}^n (\pi_i^d - P_i) \dwt(i),
\end{align}
where note that $P_i = P(Y = i),$ and $\pi_i^d = P(Y_d = i)$.
Recall from \cref{lem:pi_d} that the distribution of $\pi^d$ is:
\begin{align*}
    \forall i < n, \pi_i^d &= p_1^i p_n, \\
    \pi_n^d &= p_1^n.
\end{align*}

Rewriting the formula for $P_i$ from \cref{lem:p}, we have
\begin{align*}
    P_0 &= \mu \frac{p_n}{\mu_n}, \\
    \forall 1 \le i < n, P_i &= \mu \frac{p_1^i p_n}{i \mu_1}, \\
    P_n &= \mu \frac{p_1^n}{n \mu_1}.
\end{align*}

Note that $\pi_i^d = \frac{\mu}{\mu_i} P_i$, where $\mu_i$ is the instantaneous completion rate in state $i$. Applying (\ref{y_d}),
\begin{align}
    \E[\Delta(Y_d)] &= (p_1^n - \mu \frac{p_1^n}{n \mu_1}) \dwt(n) + \sum_{i=1}^{n-1} (p_1^i p_n - \mu \frac{p_1^i p_n}{i \mu_1}) \dwt(i) \nonumber\\
    &= p_1^n (1 - \frac{\mu}{n \mu_1}) \dwt(n) + \sum_{i=1}^{n-1} p_1^i p_n(1 - \frac{\mu}{i \mu_1}) \dwt(i). \label{delta_d}
\end{align}

Next, we calculate $\dwt(i)$. Using  the Poisson equation \eqref{eq:poisson}, which characterizes $\Delta(i)$ up to an absolute constant, we obtain
\begin{align*}
    \forall 1 \le i < n, \Delta(i) &= 1 - \frac{\mu}{i \mu_1} + \Delta(i-1) \\
    \implies \forall 1 \le i < n, \Delta(i) - \Delta(i-1) &= 1 - \frac{\mu}{i \mu_1} \\
    \implies \forall 1 \le i < n, \dwt(i) &= i - \frac{\mu}{\mu_1} H_i.
\end{align*}
Note that state $n$ is a special case:
\begin{align*}    
    \Delta(n) &= 1 - \frac{\mu}{n \mu_1} + p_1 \Delta(n) + p_n \Delta(n-1) \\
    \implies p_n \Delta(n) &= 1 - \frac{\mu}{n \mu_1} + p_n \Delta(n-1) \\
    \implies \Delta(n) - \Delta(n-1) &= \frac{1}{p_n} - \frac{\mu}{n p_n \mu_1} \\
    \implies \dwt(n) &= n-1 + \frac{1}{p_n} - \frac{\mu}{\mu_1} H_{n-1} - \frac{\mu}{n p_n \mu_1}.
\end{align*}

Plugging the expressions of $\dwt(i)$ for $i=1,\ldots,n$ into (\ref{delta_d}), we have 

\begin{align*}
    \E[\Delta(Y_d)] &=  p_1^n (1 - \frac{\mu}{n \mu_1}) ( n-1 + \frac{1}{p_n} - \frac{\mu}{\mu_1} H_{n-1} - \frac{\mu}{n p_n \mu_1}) \\
   & +  \sum_{i=1}^{n-1} p_1^i p_n(1 - \frac{\mu}{i \mu_1}) (i - \frac{\mu}{\mu_1} H_i).
   \qedhere
\end{align*}
\end{proof}

In the next three sections, we will characterize the asymptotic growth rate of the $\mu$ expression given in \cref{lem:explicit-mu} and the $\E[\Delta(Y_d)]$ expression given in \cref{lem:explicit-delta} for the three regimes discussed in \cref{sec:main}.

\section{Regime 1: $p_n = \omega(1/n)$: $n$-server dominated regime.}
\label{sec:regime-1}

This section focuses on  regime 1, namely the $n$-server dominated regime, where the load of $n$-server jobs asymptotically dominates the load of 1-server jobs.
To prove \cref{thm:n-server-dominated}, we start with \cref{lem:r1-mu} which characterizes the throughput $\mu$ in this regime.

\begin{lemma}
    \label{lem:r1-mu}
    In regime 1 which is the $n$-server dominated regime,
    namely $p_n = \omega(1/n)$, as $p_n \to 0, n \to \infty$,
    \begin{align*}
       \mu = (1+o(1))\frac{\mu_1}{p_n \ln 1/p_n}.
    \end{align*}
\end{lemma}
\begin{proof}
    From \cref{lem:explicit-mu}, we have the following expression for $\mu$:
    \begin{align*}
        \mu &= \frac{1}{p_n} \left( \frac{1}{\mu_n} + \frac{1}{n\mu_1} \frac{p_1^n}{p_n} + \sum_{b=1}^{n-1} \frac{p_1^b}{b \mu_1}\right)^{-1}.
    \end{align*}
    Note that as $n \to \infty$, $p_1^n = (1-p_n)^n$ can only converge to a nonzero constant if $p_n = \theta(1/n)$. In this setting, $p_n = \omega(1/n)$, so $p_1^n = o(1)$. Note that $\frac{p_1^n}{n p_n} = \frac{o(1)}{\omega(1)} = o(1)$. Then
\begin{align*}
    \mu &= \frac{1}{p_n} \left( \frac{1}{\mu_n} + \frac{1}{n\mu_1} \frac{p_1^n}{p_n} + \sum_{b=1}^{n-1} \frac{p_1^b}{b \mu_1}\right)^{-1} \\
    &= \frac{1}{p_n} \left( \frac{1}{\mu_n} + o(1) + \sum_{b=1}^{n-1} \frac{p_1^b}{b \mu_1}\right)^{-1} + o(1) \\
    &= (1+o(1))\frac{1}{p_n} \left( \frac{1}{\mu_n} + \sum_{b=1}^{n-1} \frac{p_1^b}{b \mu_1}\right)^{-1}.
\end{align*}

Here we use the fact that $\frac{1}{\mu_n}$ is a constant, so the $o(1)$ term has a negligible multiplicative effect. Moreover, we will show that $\sum_{b=1}^{n-1} \frac{p_1^b}{b \mu_1} = \theta(\ln 1/p_n)$, so $\frac{1}{\mu_n}$ is also multiplicatively negligible.
\begin{align}
    \label{eq:r1-mu-mid}
    (1+o(1))\frac{1}{p_n} \left( \frac{1}{\mu_n} + \sum_{b=1}^{n-1} \frac{p_1^b}{b \mu_1}\right)^{-1}
    =(1+o(1))\frac{\mu_1}{p_n} \left(\sum_{b=1}^{n-1} \frac{p_1^b}{b }\right)^{-1}.
\end{align}
We next focus on  $\sum_{b=1}^{n-1} \frac{p_1^b}{b}$. First, we change from a finite summation to an infinite summation, and second we change to an integral.
\begin{align}
    \label{eq:r1-mu-next}
    \sum_{b=1}^{n-1} \frac{p_1^b}{b}  = (1+o(1)) \sum_{b=1}^{\infty} \frac{p_1^b}{b} =(1+o(1))\int_{b=1}^\infty \frac{p_1^b}{b} db.
\end{align}

To justify the change to an infinite summation,
note that $p_n = \omega(1/n)$ and hence, $1/p_n = o(n)$.
At index $b = \lfloor 1/p_n \rfloor$, $p_1^b \sim 1/e$.
In particular, all indices from $\lfloor 1/p_n \rfloor$ to $\lfloor 2/p_n \rfloor$ sum to at most $1/(p_n e)$:
\begin{align*}
    \sum_{b=\lfloor 1/p_n \rfloor}^{\lfloor 2/p_n \rfloor} \frac{p_1^b}{b} \le \frac{1}{p_n} \frac{p_n}{e} = \frac{1}{e}.
\end{align*}
By a similar argument, indices from $\lfloor 2/p_n  \rfloor$ to $\lfloor 3/p_n  \rfloor$ sum to at most $\frac{1}{2 e^2}$, and in general for any constant $k$,
\begin{align*}
    \sum_{b=\lfloor k/p_n \rfloor}^{\lfloor (k+1)/p_n \rfloor} \frac{p_1^b}{b} \le \frac{1}{p_n} \frac{k p_n}{e^k} = \frac{k}{e^k}.
\end{align*}
Summing this infinite series, we find that
\begin{align*}
    \sum_{b=\lfloor 1/p_n \rfloor}^{\infty} \frac{p_1^b}{b} \le \sum_{k=1}^\infty \frac{k}{e^k} = \frac{e}{(e-1)^2}.
\end{align*}
Thus, the sum over indices $b \ge \lfloor 1/p_n \rfloor$ is $O(1)$.
As a result, the sum over indices $b \ge n$ must also be $O(1)$, as $n$ is much larger than $1/p_n$.
This $O(1)$ difference is negligible compared to the $\theta(\ln 1/p_n)$ growth rate of the integral which we will demonstrate.

To justify the transformation from the summation to the integral, note that the difference
can be rewritten as follows:
\begin{align*}
    \sum_{b=1}^{\infty} \frac{p_1^b}{b} - \int_{b=1}^\infty \frac{p_1^b}{b} db
    = \sum_{b=1}^\infty \left(\frac{p_1^b}{b} - \int_{x=b}^{b+1} \frac{p_1^x}{x} dx\right).
\end{align*}

Next, note that $\frac{p_1^b}{b}$ is decreasing as a function of $b$, so we can upper bound the difference:
\begin{align*}
    \sum_{b=1}^\infty \left(\frac{p_1^b}{b} - \int_{x=b}^{b+1} \frac{p_1^x}{x} dx\right)
    &\le \sum_{b=1}^\infty \left(\frac{p_1^b}{b} - \frac{p_1^{b+1}}{b+1}\right).
\end{align*}

Now, we can upper bound the term inside the summation:
\begin{align*}
    \frac{p_1^b}{b} - \frac{p_1^{b+1}}{b+1}
    &\le \frac{p_1^b}{b} - \frac{p_1^{b+1}}{b}
    \le p_1^b - p_1^{b+1} \\
    \implies \sum_{b=1}^\infty \left(\frac{p_1^b}{b} - \frac{p_1^{b+1}}{b+1}\right)
    &\le \sum_{b=1}^\infty \left(p_1^b - p_1^{b+1}\right) = p_1 < 1.
\end{align*}

Thus, this difference is negligible compared to the $\theta(\ln 1/p_n)$ growth rate of the integral which we will demonstrate.

We now continue by simplifying the integral from \eqref{eq:r1-mu-next},
\begin{align*}
    \int_{b=1}^\infty \frac{p_1^b}{b} db =Gamma[0, - \ln(p_1)] = Gamma[0, - \ln(1-p_n)],
\end{align*}
where Gamma is the incomplete gamma function.
Note that
\begin{align*}
    \lim_{p_n \to 0} Gamma[0, - \ln(1-p_n)] + \ln p_n = \gamma,
\end{align*}
where $\gamma \approx 0.57726$ is the Euler-Mascheroni constant.
Thus,
\begin{align}
    \label{eq:r1-incomplete-gamma}
    \int_{b=1}^\infty \frac{p_1^b}{b} db = - \ln p_n + O(1) = \ln \frac{1}{p_n} + O(1).
\end{align}
We have demonstrated the desired growth rate for the summation, and plugging that growth rate into \eqref{eq:r1-mu-mid} completes the proof.
\end{proof}
We next characterize the limiting scaled response time $\E[\Delta(Y_d)]$, given in \eqref{eq:delta-limiting-scaled}, in this regime, to complete the proof of \cref{thm:n-server-dominated}. 
\begin{theorem}
    \label{thm:r1-delta}
  In regime 1, which is the $n$-server dominated regime, namely $p_n = \omega(1/n)$, as $p_n \to 0, n \to \infty$, then   
  \begin{align*}
    \E[\Delta(Y_d)] = (1+o(1))\frac{1}{2 p_n}.
\end{align*}
\end{theorem}
\begin{proof}
    First, we have the explicit formula for $\E[\Delta(Y_d)],$ from \cref{lem:explicit-delta},
    \begin{align}
        \E[\Delta(Y_d)] &=  p_1^n (1 - \frac{\mu}{n \mu_1}) ( n-1 + \frac{1}{p_n} - \frac{\mu}{\mu_1} H_{n-1} - \frac{\mu}{n p_n \mu_1}) \nonumber\\
   & +  \sum_{i=1}^{n-1} p_1^i p_n(1 - \frac{\mu}{i \mu_1}) (i - \frac{\mu}{\mu_1} H_i).
   \label{eq:r1-delta-first}
    \end{align}

    We first rewrite expression \eqref{eq:r1-delta-first} in terms of the difference of two infinite summations:
    \begin{align}
        \E[\Delta(Y_d)] &=  p_1^n (1 - \frac{\mu}{n \mu_1}) ( n-1 + \frac{1}{p_n} - \frac{\mu}{\mu_1} H_{n-1} - \frac{\mu}{n p_n \mu_1}) \label{eq:r1-non-summation}\\
   & +  \sum_{i=1}^\infty p_1^i p_n(1 - \frac{\mu}{i \mu_1}) (i - \frac{\mu}{\mu_1} H_i) \label{eq:r1-main}\\
   &- \sum_{i=n}^\infty p_1^i p_n(1 - \frac{\mu}{i \mu_1}) (i - \frac{\mu}{\mu_1} H_i). \label{eq:r1-tail}
    \end{align}

    We will show that \eqref{eq:r1-non-summation} and \eqref{eq:r1-tail} are each $o(1/p_n)$, and thus negligible.
    We will then split up \eqref{eq:r1-main} into four terms:
    \begin{align}
        \label{eq:r1-p1i-i}
         &p_n \sum_{i=1}^\infty p_1^i i \\
        \label{eq:r1-p1i}
        - &\frac{\mu p_n}{\mu_1} \sum_{i=1}^\infty p_1^i \\
        \label{eq:r1-p1i-hi}
        - &\frac{\mu p_n}{\mu_1} \sum_{i=1}^\infty p_1^i H_i \\
        \label{eq:r1-p1i-hi-i}
        + &\frac{\mu^2 p_n}{\mu_1^2} \sum_{i=1}^\infty p_1^i \frac{1}{i} H_i
    \end{align}
    and show that the asymptotic growth rates of these terms are:
        \begin{enumerate}
        \item $p_1^i i$ term \eqref{eq:r1-p1i-i}: $\frac{1}{p_n} + o(1/p_n)$. 
        \label{it:r1-1}
        \item $p_1^i$ term \eqref{eq:r1-p1i}: $o(1/p_n)$.
        \label{it:r1-2}
        \item $p_1^i H_i$ term \eqref{eq:r1-p1i-hi}: $-\frac{1}{p_n} + o(1/p_n)$.
        \label{it:r1-3}
        \item $p_1^i H_i / i$ term \eqref{eq:r1-p1i-hi-i}:  $\frac{1}{2p_n} + o(1/p_n)$.
        \label{it:r1-4}
    \end{enumerate}
    These growth rates sum to $\frac{1}{2 p_n} + o(1/p_n)$, as desired.

    First, we will show that \eqref{eq:r1-non-summation} is negligible:
    \begin{align*}
    &p_1^n (1 - \frac{\mu}{n \mu_1}) ( n-1 + \frac{1}{p_n} - \frac{\mu}{\mu_1} H_{n-1} - \frac{\mu}{n p_n \mu_1}) \\
    =& \theta(p_1^n ( n - \frac{\mu}{n p_n})) \\
    =& \theta(p_1^n ( n - o(\frac{1}{p_n\log n}))) \\
    =& \theta(n p_1^n).
\end{align*}

Next we verify that $n p_1^n = o(1/p_n)$ which is
equivalent to showing that $n p_n = o(1/p_1^n) = o((\frac{1}{1-p_n})^n)$.

Note that $\frac{1}{1-p_n} = 1 + p_n + p_n^2 + \ldots \ge 1 + p_n$.
Thus,
\begin{align*}
    \left(\frac{1}{1-p_n}\right)^n \ge (1 + p_n)^n \ge p_n^2 \frac{n(n-1)}{2} = \theta(p_n^2 n^2).
\end{align*}
Note that in regime 1, $np_n = \omega(1)$. Thus, $p_n^2 n^2 = \omega(n p_n)$, and equivalently $n p_n = o(1/p_1^n)$, as desired and \eqref{eq:r1-non-summation} is negligible.

Next, we will show that \eqref{eq:r1-tail} is negligible:
\begin{align*}
    &\sum_{i=n}^\infty p_1^i p_n(1 - \frac{\mu}{i \mu_1}) (i - \frac{\mu}{\mu_1} H_i).
\end{align*}

Let us start by bounding the term corresponding to $i=n$ of \eqref{eq:r1-tail}, namely,
\begin{align*}
    p_1^n p_n(1 - \frac{\mu}{n \mu_1}) (n - \frac{\mu}{\mu_1} H_n).
\end{align*}

Recall that we have shown that $p_1^n = o(\frac{1}{n p_n})$. Then

\begin{align*}
    &p_1^n p_n(1 - \frac{\mu}{n \mu_1}) (n - \frac{\mu}{\mu_1} H_n) \\
    =& o(1/n)(1 - \frac{\mu}{n \mu_1}) (n - \frac{\mu}{\mu_1} H_n) \\
    =& o(1/n)(1 - \frac{1}{n p_n \log 1/p_n}) (n - \frac{\log n}{p_n \log 1/p_n}) \\
    =& o(1/n) (n - \frac{\log n}{p_n \log 1/p_n}) \\
    =& o(1) - o(\frac{\log n}{n p_n \log 1/p_n}).
\end{align*}

Note that the function $x \ln 1/x$ is increasing in $x$ in the interval $(0, 1/e)$.
Recall that $p_n =\omega(1/n)$.
Thus, $p_n \log 1/p_n = \Omega(1/n \log n).$
We therefore find that
\begin{align*}
    \frac{\log n}{n p_n \log 1/p_n}
    = O(\frac{\log n}{n (1/n \log n)}) = O(1).
\end{align*}

As a result, we find that the term corresponding to $i=n$  of \eqref{eq:r1-tail} is $o(1)$.

Next, consider the first $\lfloor 1/p_n \rfloor$ terms of \eqref{eq:r1-tail}, ranging from $i=n$ to $i=n+\lfloor 1/p_n \rfloor$.
By the same argument, each of these terms is $o(1)$, and the sum is bounded accordingly:
\begin{align*}
    \sum_{i=n}^{i=n + \lfloor 1/p_n \rfloor} p_1^i p_n(1 - \frac{\mu}{i \mu_1}) (i - \frac{\mu}{\mu_1} H_i) = o(1/p_n).
\end{align*}

Call these $\lfloor 1/p_n \rfloor$ terms the first \emph{block} of terms, and similarly, define the $k$th block of terms:
\begin{align*}
    b(k) := \sum_{i=n+ \lfloor (k-1)/p_n \rfloor}^{i=n+ \lfloor k/p_n \rfloor} p_1^i p_n(1 - \frac{\mu}{i \mu_1}) (i - \frac{\mu}{\mu_1} H_i).
\end{align*}

We can bound the $k$th block's ratio with respect to the first term of \eqref{eq:r1-tail}. The $p_1^i$ term has shrunk by a factor of $p_1^{k/p_n}$, which is $(1+o(1))e^{-k}$. The  $1-\frac{\mu}{i\color{blue}\mu_1}$ term differs negligibly from 1 and the $i - \frac{\mu}{\mu_1} H_i$ term is $O(i)$ as shown above. As a result, the ratio between this term in the $k$th block and the first block is at most $k$. Thus, we can bound the ratio of the blocks:
\begin{align*}
    b(k)/b(1) &\le e^{-k}(k+1) \\
    \sum_{k=1}^\infty b(k) &\le b(1) \sum_{k=1}^\infty e^{-k} (k+1) \le 3 b(1) = o(1/p_n).
\end{align*}
As a result, the infinite sum in \eqref{eq:r1-tail} is $o(1/p_n)$, as desired.

Now, we're ready to characterize the four subterms of the main term \eqref{eq:r1-main}, starting with \eqref{eq:r1-p1i-i}, confirming \cref{it:r1-1}:
\begin{align*}
    p_n \sum_{i=1}^\infty p_1^i i = p_n \frac{p_1}{p_n^2} = \frac{p_1}{p_n} = \frac{1}{p_n} + o(1/p_n).
\end{align*}

Next, we have \eqref{eq:r1-p1i}, confirming \cref{it:r1-2}:
\begin{align*}
    - &\frac{\mu p_n}{\mu_1} \sum_{i=1}^\infty p_1^i =  - \frac{\mu p_n}{\mu_1} \frac{p_1}{p_n} = - \frac{\mu p_1}{\mu_1} = \theta(\mu) = \theta(\frac{1}{p_n \log(1/p_n)} )= o(1/p_n).
\end{align*}

Next, we consider \eqref{eq:r1-p1i-hi}, in which we replace $H_i$ with $\ln i + O(1)$:
\begin{align*}
    - &\frac{\mu p_n}{\mu_1} \sum_{i=1}^\infty p_1^i H_i = 
    - \frac{\mu p_n}{\mu_1} \sum_{i=1}^\infty p_1^i (\ln i + O(1)).
\end{align*}
We next switch from a summation to an integral, and bound the difference from this change:
\begin{align*}
    \sum_{i=1}^\infty p_1^i \ln i - \int_{i=1}^\infty p_1^i \ln i di
    = \sum_{i=1}^\infty \left( p_1^i \ln i - \int_{j=i}^{i+1} p_1^j \ln j dj \right).
\end{align*}

Because $p_1^j$ is decreasing in $j$ and $\ln j$ is increasing in $j$, we can upper bound this difference as
\begin{align*}
    \sum_{i=1}^\infty \left( p_1^i \ln i - \int_{j=i}^{i+1} p_1^j \ln j dj \right)
    &\le \sum_{i=1}^\infty \left( p_1^i \ln i -  p_1^{i+1} \ln i \right)
    = \sum_{i=1}^\infty \left( (p_1^i - p_1^{i+1}) \ln i \right) \\
    &= (1 - p_1)\sum_{i=1}^\infty \left( p_1^i  \ln i \right).
\end{align*}

We have  shown that the difference between the summation and the integral is equal to the summation multiplied by a $1-p_1 = o(1)$ term.
We can thus relate \eqref{eq:r1-p1i-hi} to the integral as follows:
\begin{align*}
    - \frac{\mu p_n}{\mu_1} \sum_{i=1}^\infty p_1^i H_i
    &=- \frac{\mu p_n}{\mu_1} \sum_{i=1}^\infty p_1^i (\ln i + O(1))
    = - \frac{\mu p_n}{\mu_1} \left( O(\frac{1}{p_n}) + \sum_{i=1}^\infty p_1^i \ln i \right) \\
    &= - \frac{\mu p_n}{\mu_1} \left( O(\frac{1}{p_n}) + (1+o(1)) \int_{i=1}^\infty p_1^i \ln i di \right).
\end{align*}

Now, let's focus on the integral:
\begin{align*}
    \int_{i=1}^\infty p_1^i \ln(i) di &= \frac{LogIntegral[p_1]}{\ln(p_1)},
\end{align*}
where LogIntegral is the logarithmic integral function $\int_0^x \frac{dt}{\ln t}$.

We can take the following limits:
\begin{align}
    \label{eq:logintegral}
    \lim_{p_n \to 0^+} LogIntegral[1-p_n] - \ln p_n &= \gamma \implies LogIntegral[1-p_n] =  \ln p_n + O(1)\\
    \label{eq:log-p1}
    \lim_{p_n \to 0^+} \frac{\ln(1-p_n)}{p_n} &= -1  \implies \ln(1-p_n) = -p_n + o(p_n)\\
    \nonumber
    \implies \frac{LogIntegral[p_1]}{\ln(p_1)} &= \frac{\ln p_n + O(1)}{-p_n + o(p_n)} = -(1+o(1)) \frac{\ln p_n}{p_n}.
\end{align}

Combining it all together, we get 
\begin{align*}    
    - \frac{\mu p_n}{\mu_1} \sum_{i=1}^\infty p_1^i H_i &= - \frac{\mu p_n}{\mu_1} \left(-(1+o(1)) \frac{\ln p_n}{p_n} + O(1/p_n)\right) \\
        &= (1+o(1)) \frac{1}{\ln(1/p_n)} \frac{\ln(p_n)}{p_n} \\
    &= -(1+o(1))\frac{1}{p_n}.
\end{align*}
This confirms \cref{it:r1-3}, the specified growth rate of \eqref{eq:r1-p1i-hi}.

We now finish with \eqref{eq:r1-p1i-hi-i}. We start by transforming $H_i$ to $\ln i$ and the summation to an integral in the same fashion as for \eqref{eq:r1-p1i}, with the same error bound:
\begin{align*}
    \frac{\mu^2 p_n}{\mu_1^2} \sum_{i=1}^\infty p_1^i \frac{1}{i} H_i
    = \frac{\mu^2 p_n}{\mu_1^2} \left((1+o(1))\int_{i=1}^\infty p_1^i \frac{1}{i} \ln i di + O(1/p_n)\right).
\end{align*}
We evaluate this integral  using a symbolic algebraic package (specifically, Mathematica \cite{mathematica}) and obtain
\begin{align}
    \label{eq:ln-i-i-infty}
    \int_{i=1}^\infty p_1^i \ln(i)/i di &= \frac{1}{2} \ln(1/p_n)^2 + O(\ln(1/p_n)).
\end{align}

Putting it all together:
\begin{align*}
        p_n \frac{\mu^2}{\mu_1^2} \sum_{i=1}^\infty p_1^i  \frac{H_i}{i} 
    &= p_n \frac{\mu^2}{\mu_1^2} \left( \frac{1}{2} \ln(1/p_n)^2 + O(\ln(1/p_n))\right) \\
    &= (1+o(1)) \frac{1}{p_n \ln(1/p_n)^2}\frac{1}{2} \ln(1/p_n)^2 \\
    &= (1+o(1))\frac{1}{2} \frac{1}{p_n}.
\end{align*}

This confirms \cref{it:r1-4}, the desired growth rate of \eqref{eq:r1-p1i-hi-i} and completes the proof.
\end{proof}
\section{Regime 2: $p_n = c/n$: Balanced load}
\label{sec:regime-2}

This section focuses on regime 2, namely, the balanced regime, where the load of $n$-server jobs and $1$-server jobs remain in fixed proportion, for a general proportionality constant $c$.
\textcolor{black}{To prove the result in \cref{thm:n-server-dominated} for regime 2, we start by characterizing the throughput $\mu$ in this regime.
\cref{lem:r2-mu} proves the first equation of \cref{thm:n-server-dominated}}.

\begin{lemma}
    \label{lem:r2-mu}
    In regime 2, which is the balanced regime, namely $p_n = c/n$, for some positive constant $c$, as $p_n \to 0, n \to \infty$,
    \begin{align*}
        \mu = (1 + o(1)) \frac{\mu_1}{p_n \ln 1/p_n}.
    \end{align*}
\end{lemma}

\begin{proof}
    From \cref{lem:explicit-mu}, we have the following expression for $\mu$:
    \begin{align*}
        \mu &= \frac{1}{p_n} \left( \frac{1}{\mu_n} + \frac{1}{n\mu_1} \frac{p_1^n}{p_n} + \sum_{b=1}^{n-1} \frac{p_1^b}{b \mu_1}\right)^{-1}.
    \end{align*}
    Note that as $n \to \infty$, $p_1^n = (1-p_n)^n = (1+o(1))e^{-c}$. Then
\begin{align}
    \nonumber
    \mu &= \frac{1}{p_n} \left( \frac{1}{\mu_n} + (1+o(1)) \frac{e^{-c}}{c\mu_1} + \sum_{b=1}^{n-1} \frac{p_1^b}{b \mu_1}\right)^{-1} \\
    \label{eq:r2-mu-mid}
    &= \frac{1}{p_n} \left( O(1) + \sum_{b=1}^{n-1} \frac{p_1^b}{b \mu_1}\right)^{-1}.
\end{align}
We will show that $\sum_{b=1}^{n-1} \frac{p_1^b}{b} = \omega(1)$, and hence that the $O(1)$ term is negligible.

To do so, we first change from a finite summation to an infinite summation, and then change to an integral:
\begin{align*}
    \sum_{b=1}^{n-1} \frac{p_1^b}{b}  = (1+o(1)) \sum_{b=1}^{\infty} \frac{p_1^b}{b} =(1+o(1))\int_{b=1}^\infty \frac{p_1^b}{b} db.
    \end{align*}

To justify the change to an infinite summation,
note that $p_n = c/n$ and hence, $1/p_n = n/c$.
At index $b = \lfloor 1/p_n \rfloor$, $p_1^b \approx 1/e,$ and $p_1^b/b \approx p_n/e$.
All indices from $\lfloor 1/p_n \rfloor$ to $\lfloor 2/p_n \rfloor$ sum to at most $1/e$,
and indices from $\lfloor 2/p_n \rfloor$ to $\lfloor 3/p_n\rfloor$ sum to at most $2/e^2$.
In general,
\begin{align*}
    \sum_{b=\lfloor k/p_n \rfloor}^{\lfloor (k+1)/p_n \rfloor} \frac{p_1^b}{b} &\le \frac{k}{e^k} \\
    \implies \sum_{b=\lfloor 1/p_n \rfloor}^\infty \frac{p_1^b}{b} &\le \sum_{k=1}^\infty \frac{k}{e^k} = \frac{e}{(e-1)^2} = O(1).
\end{align*}
As $1/p_n = \theta(n)$, we similarly find that the sum from index $n$ to $\infty$ is $O(1)$.
As stated above, we will show that this $O(1)$ term is negligible.

The transformation from the summation to the integral holds for the same reason as given in \cref{lem:r1-mu}, again with a negligible $O(1)$ difference.

Next, as shown in \eqref{eq:r1-incomplete-gamma}, we have
\begin{align*}
    \int_{b=1}^\infty \frac{p_1^b}{b} db= -\ln p_n + O(1).
\end{align*}
Note that the above integral is thus $\omega(1)$, so the prior $O(p_n)$ term is negligible as specified above. Substituting everything into \eqref{eq:r2-mu-mid}, we have
\begin{align*}
    \mu = \frac{1}{p_n} \left( O(1) - \ln p_n\right)^{-1}
     = (1+o(1))\frac{1}{p_n \ln 1/p_n},
\end{align*}
as desired.
\end{proof}

\textcolor{black}{We next characterize the limiting scaled response time $\E[\Delta(Y_d)]$ for regime 2 to complete the proof of \cref{thm:n-server-dominated} in the balanced regime}. 

\begin{theorem}
    \label{thm:r2-delta}
    In regime 2, which is the balanced regime, namely $p_n = c/n$, for some positive constant $c$, as $p_n \to 0, n \to \infty$,
    \begin{align*}
        \E[\Delta(Y_d)] = (1+o(1)) \frac{1}{2 p_n}.
    \end{align*}
\end{theorem}
\begin{proof}
    We start with the formula for $\E[\Delta(Y_d)]$, from \cref{lem:explicit-delta}, which we split into five terms:
    \begin{align}
        \label{eq:r2-non-summation}
        &p_1^n (1 - \frac{\mu}{n \mu_1}) ( n-1 + \frac{1}{p_n} - \frac{\mu}{\mu_1} H_{n-1} - \frac{\mu}{n p_n \mu_1}) \\
        \label{eq:r2-p1i-i}
        + &p_n \sum_{i=1}^{n-1} p_1^i i \\
        \label{eq:r2-p1i}
        - &\frac{\mu p_n}{\mu_1} \sum_{i=1}^{n-1} p_1^i \\
        \label{eq:r2-p1i-hi}
        - &\frac{\mu p_n}{\mu_1} \sum_{i=1}^{n-1} p_1^i H_i \\
        \label{eq:r2-p1i-hi-i}
        + &\frac{\mu^2 p_n}{\mu_1^2} \sum_{i=1}^{n-1} p_1^i \frac{1}{i} H_i
    \end{align}
    We will demonstrate the following asymptotic growth rates for each of these terms:
    \begin{enumerate}
        \item Non-summation term \eqref{eq:r2-non-summation}: $\frac{c e^{-c}}{p_n} + o(1/p_n)$.
        \label{it:r2-1}
        \item $p_1^i i$ term \eqref{eq:r2-p1i-i}: $\frac{1 - e^{-c}(1+c)}{p_n} + o(1/p_n)$.
        \label{it:r2-2}
        \item $p_1^i$ term \eqref{eq:r2-p1i}: $o(1/p_n)$.
        \label{it:r2-3}
        \item $p_1^i H_i$ term \eqref{eq:r2-p1i-hi}: $- \frac{1-e^{-c}}{p_n} + o(1/p_n)$.
        \label{it:r2-4}
        \item $p_1^i H_i / i$ term \eqref{eq:r2-p1i-hi-i}: $\frac{1}{2 p_n} + o(1/p_n)$.
        \label{it:r2-5}
    \end{enumerate}
    Note that the non-summation term \eqref{eq:r2-non-summation} is now significant, which was not the case in the $n$-server dominated regime in \cref{sec:regime-1}.
    
    These growth rates sum to $\frac{1}{2 p_n} + o(1/p_n)$, as desired. Furthermore, $c$ does not appear in the result, as its contributions cancel out.

    Now, to verify \cref{it:r2-1,it:r2-2,it:r2-3,it:r2-4,it:r2-5}, we start with the non-summation term \eqref{eq:r2-non-summation}, and simplify. First, we note that $p_1^n \to e^{-c}$:
\begin{align*}
    &p_1^n (1 - \frac{\mu}{n \mu_1}) ( n-1 + \frac{1}{p_n} - \frac{\mu}{\mu_1} H_{n-1} - \frac{\mu}{n p_n \mu_1}) \\
    &= (1+o(1))e^{-c}(1 - \frac{\mu}{n \mu_1}) ( n-1 + \frac{1}{p_n} - \frac{\mu}{\mu_1} H_{n-1} - \frac{\mu}{n p_n \mu_1}).
\end{align*}
Next, we remove negligible terms.
Note that $\mu/n$ is negligible, because $\mu = O(\frac{1}{p_n \ln(1/p_n)})$, and $n = O(1/p_n)$.
\begin{align*}
    &(1+o(1))e^{-c}(1 - \frac{\mu}{n \mu_1}) ( n-1 + \frac{1}{p_n} - \frac{\mu}{\mu_1} H_{n-1} - \frac{\mu}{n p_n \mu_1}) \\
    &= (1+o(1))e^{-c} ( \frac{c + 1}{p_n} - \frac{\mu}{\mu_1} H_{n-1}).
\end{align*}
We now apply \cref{lem:r2-mu}, and use the facts that that $|H_{n-1} - \ln n| \le 1$ by standard harmonic series bounds, and $\ln n - \ln 1/p_n = \ln c$. Thus,
\begin{align*}
    &(1+o(1))e^{-c} ( \frac{c + 1}{p_n} - \frac{\mu}{\mu_1} H_{n-1}) = (1+o(1))e^{-c} ( \frac{c}{p_n}).
\end{align*}
This verifies \cref{it:r2-1}, the claimed asymptotic growth  rate for \eqref{eq:r2-non-summation}.

Next, we examine \cref{it:r2-2}, the claimed asymptotic growth rate for \eqref{eq:r2-p1i-i}. This summation has a simple closed form:
\begin{align*}
    p_n \sum_{i=1}^{n-1} p_1^i i &=
    p_n \frac{p_1+p_1^{n+1}(-1-n+n p_1)}{p_n^2}
    = (1+o(1))\frac{1+e^{-c}(-1-n(1-p_1))}{p_n} \\
    &= (1+o(1))\frac{1-e^{-c}(1+c)}{p_n}.
\end{align*}

Next, we examine \cref{it:r2-3}, the claimed asymptotic growth rate for \eqref{eq:r2-p1i}. This summation also has a simple closed form:
\begin{align*}
    \sum_{i=1}^{n-1} p_1^i &= \frac{1-p_1^n}{1-p_1} = \frac{1-e^{-c}}{p_n} \\
\implies    -\frac{\mu p_n}{\mu_1} \sum_{i=1}^{n-1} p_1^i &= -(1+o(1))\frac{1-e^{-c}}{p_n \ln(1/p_n)} = -(1+o(1))\frac{1-e^{-c}}{p_n \ln(1/p_n)}
\end{align*}

Note that $\frac{1}{p_n \ln (1/p_n)} = o(1/p_n)$, confirming the growth rate given in \cref{it:r2-3}.

Next, we examine \cref{it:r2-4}, the claimed asymptotic growth rate for \eqref{eq:r2-p1i-hi}:
\begin{align*}
    - &\frac{\mu p_n}{\mu_1} \sum_{i=1}^{n-1} p_1^i H_i.
\end{align*}

We claim that replacing $H_i$ by $\ln(i)$ has a negligible impact on the sum, as does transforming the sum into an integral.
In particular, $H_i - \ln(i) \in [0, 1]$. Similar to  the argument in \cref{lem:r1-mu}, replacing the sum by an integral changes the growth rate by at most a bounded constant per term.
Both have at most an additive impact per term, resulting in net change on the order of $\frac{\mu p_n}{\mu_1} \sum_{i=1}^{n-1} p_1^i$, which matches the formula from \eqref{eq:r2-p1i}, and is similarly negligible. As a result, we have an integral that can be evaluated explicitly. Thus,
\begin{align}
    \label{eq:r2-replace}
    \sum_{i=1}^{n-1} p_1^i H_i &= (1+o(1)) \int_{i=1}^n{p_1^i \ln(i)} di \\
    \nonumber
    \int_{i=1}^n p_1^i \ln i di &= \frac{-ExpIntegralEi[n \ln p_1] + p_1^n \ln n + LogIntegral[p_1]}{\ln p_1},
\end{align}
where $ExpIntegralEi[x] := -\int_{-x}^\infty \frac{e^{-t}}{t} dt$ is the exponential integral function.

Taking limits, we find that
\begin{align*}
    \lim_{n \to \infty} ExpIntegralEi[n \ln (1-c/n)] &= ExpIntegralEi[c]\\
    \text{equivalently, } ExpIntegralEi[n \ln p_1] &= \theta(1) 
\end{align*}
Recall the asymptotic expression from \cref{sec:regime-1}, in equations \eqref{eq:logintegral} and \eqref{eq:log-p1}.
There, we used the following limits:
\begin{align*}
    LogIntegral[p_1] &= \ln p_n + \theta(1) = -\ln n + \theta(1), \\
    \ln p_1 &= -p_n(1+o(1)).
\end{align*}

It is also simple to calculate that:
\begin{align*}
    p_1^n \ln n &= e^{-c} \ln n + o(\ln n), \\
    \frac{\mu p_n}{\mu_1} &= (1+o(1)) \frac{1}{\ln 1/p_n} = \frac{1}{\ln n} + o(1/\ln n).
\end{align*}

As a result,
\begin{align*}
    -ExpIntegralEi[n \ln p_1] + p_1^n \ln n + LogIntegral[p_1] &= (e^{-c} - 1) \ln n + o(\ln n) \\
    \implies \sum_{i=1}^{n-1} p_1^i H_i &= (1+o(1))\frac{(1 - e^{-c}) \ln n}{p_n} \\
    \implies \frac{\mu p_n}{\mu_1} \sum_{i=1}^{n-1} p_1^i H_i &= \frac{1 - e^{-c}}{p_n} + o(1/p_n).
\end{align*}

This confirms \cref{it:r2-4},  claimed asymptotic growth rate for \eqref{eq:r2-p1i-hi}.

Finally, we have \cref{it:r2-5}, the claimed asymptotic growth rate for \eqref{eq:r2-p1i-hi-i}:
\begin{align*}
    \frac{\mu^2 p_n}{\mu_1^2} \sum_{i=1}^{n-1} p_1^i \frac{1}{i} H_i.
\end{align*}
Let us focus on the leading multiplier first:
\begin{align*}
    \frac{\mu^2 p_n}{\mu_1^2} &= \frac{1}{p_n \ln(1/p_n)^2}.
\end{align*}

Now, we turn to the summation. Note that a component of the summation must be of the order of $\ln(1/p_n)^2$ to be non-negligible.

We perform the same transformation as in \eqref{eq:r2-replace},
replacing the summation by an integral and the $H_{i}$ with $\ln i$, with the same justification:
\begin{align*}
    \sum_{i=1}^{n-1} p_1^i \frac{H_i}{i}
    \approx \int_{i=1}^n p_1^i \frac{H_i}{i} di
    =\int_{i=1}^{c/p_n} (1-p_n)^i \frac{H_i}{i} di
    \approx \int_{i=1}^{c/p_n} (1-p_n)^i \frac{\ln i}{i} di.
\end{align*}

Note that this differs from \eqref{eq:ln-i-i-infty} because the upper limit of the integral is $n = c/p_n$, not $\infty$.
Nonetheless, using Mathematica \cite{mathematica}, we obtain the same asymptotics:

\begin{align*}
    \int_{i=1}^{c/p_n} (1-p_n)^i \frac{\ln i}{i} di = \frac{1}{2} \ln(1/p_n)^2 + \theta(\log(1/p_n))
\end{align*}


Thus the entire term satisfies the claim:
\begin{align*}
    \frac{\mu^2 p_n}{\mu_1^2} \sum_{i=1}^{n-1} p_1^i \frac{1}{i} H_i
    &= \frac{1}{2 p_n} + o(1/p_n).
\end{align*}

All five claims, namely \cref{it:r2-1,it:r2-2,it:r2-3,it:r2-4,it:r2-5}, are verified, completing the proof.
\end{proof}

\section{Regime 3: $p_n = o(1/n)$: Load dominated by 1-server jobs}
\label{sec:regime-3}
This section focuses on regime 3, namely, the 1-server dominated regime, where the load of $1$-server jobs asymptotically overwhelms the load of $n$-server jobs.
We start in \cref{sec:r3-gen} in the general setting, covering all asymptotic relationships for $p_n = o(1/n)$.
However, in \cref{sec:r3-spec}, we will primarily focus on a setting of polynomial scaling, where $p_n = 1/n^\alpha, \alpha > 1$.

\subsection{Regime 3, general setting: Any growth rate $p_n = o(1/n)$}
\label{sec:r3-gen}
This setting contains a wide variety of behaviors, so while we characterize asymptotic behavior of the throughput $\mu$ in significant detail,
our result on scaled mean queue length, $\E[\Delta(Y_d)]$, is less tight. Our tighter results are in the more specific setting of \cref{sec:r3-spec}.

We start by characterizing the throughput $\mu$ in this regime.
\begin{lemma}
    \label{lem:r3-gen-mu}
    If $p_n = o(1/n \ln n)$, then
    \begin{align*}
        \mu = (1+o(1)) n \mu_1.
    \end{align*}
    On the other hand, if $p_n = \omega(1/(n \ln n))$ and $p_n = o(1/n)$,
    then
    \begin{align*}
        \mu = (1+o(1)) \frac{\mu_1}{p_n \ln n}.
    \end{align*}
\end{lemma}
\begin{proof}
    We start with \cref{lem:explicit-mu}:
\begin{align*}
    \mu &= \frac{1}{p_n} \left( \frac{1}{\mu_n} + \frac{1}{n\mu_1} \frac{p_1^n}{p_n} + \sum_{b=1}^{n-1} \frac{p_1^b}{b \mu_1}\right)^{-1} \\
    &= \frac{1}{p_n} \left( \frac{1}{\mu_n} + (1+o(1))\frac{1}{n \mu_1 p_n} + \sum_{b=1}^{n-1} \frac{p_1^b}{b \mu_1}\right)^{-1}.
\end{align*}
This step holds because $p_n = o(1/n)$, so $p_1^n = 1+o(1)$.
As a result, we can similarly simplify the summation term -- note that $p_1^n \le p_1^b \le 1$. Furthermore, $n p_n = o(1)$, so the $1/\mu_n$  term is negligible. Then
\begin{align*}
    \mu &= \frac{1+o(1)}{p_n} \left(\frac{1}{n \mu_1 p_n} + \sum_{b=1}^{n-1} \frac{1}{b \mu_1}\right)^{-1} 
    = \frac{1+o(1)}{p_n} \left(\frac{1}{n \mu_1 p_n} + \frac{H_{n-1}}{\mu_1}\right)^{-1} \\
    &= \frac{1+o(1)}{p_n} \left( \frac{1}{n \mu_1 p_n} + \frac{\ln n}{\mu_1}\right)^{-1} 
\end{align*}

Either of these two terms can be larger in asymptotic limit,
depending on the asymptotic behavior of $p_n$.
If $p_n = o(1/(n \ln n))$, then $\frac{1}{n \mu_1 p_n}$ is larger.
If $p_n = \omega(1/(n \ln n))$, then $\frac{\ln n}{\mu_1}$ is larger.
In the former case,
\begin{align*}
    \mu &= \frac{1+o(1)}{p_n} \left( \frac{1}{n \mu_1 p_n}\right)^{-1} = (1+o(1)) \frac{n \mu_1 p_n}{p_n} = (1+o(1)) n \mu_1.
\end{align*}
In the latter case,
\begin{align*}
    \mu &= \frac{1+o(1)}{p_n} \left(\frac{\ln n}{\mu_1}\right)^{-1}
    = (1+o(1)) \frac{\mu_1}{p_n \ln n}.
\end{align*}
Finally, if $p_n = \theta(1/(n \ln n))$, then both terms $\frac{1}{n \mu_1 p_n}$ and $\frac{\ln n}{\mu_1}$
are within a constant factor of each other in the asymptotic limit, 
and both contribute to the asymptotic behavior of $\mu$.
\end{proof}

We now switch to $\E[\Delta(Y_d)]$, where we only focus on the case where $p_n = o(1/(n \ln n))$,
and only provide an upper bound on the asymptotic growth of $\E[\Delta(Y_d)]$. Our stronger results for $\E[\Delta(Y_d)]$ are in  \cref{sec:r3-spec}.

\begin{lemma}
    \label{lem:r3-gen-delta}
    If $p_n = o(1/(n \ln n))$, then
    \begin{align*}
        \E[\Delta(Y_d)] = o(1/p_n).
    \end{align*}
\end{lemma}
\begin{proof}
    We start with the expression of $\E[\Delta(Y_d)]$ from \cref{lem:explicit-delta}, which we separate into two terms.
    \begin{align}
    \label{eq:r3-gen-nonsummation}
    \E[\Delta(Y_d)] &=  p_1^n (1 - \frac{\mu}{n \mu_1}) ( n-1 + \frac{1}{p_n} - \frac{\mu}{\mu_1} H_{n-1} - \frac{\mu}{n p_n \mu_1}) \\
    \label{eq:r3-gen-summation}
    &+  \sum_{i=1}^{n-1} p_1^i p_n(1 - \frac{\mu}{i \mu_1}) (i - \frac{\mu}{\mu_1} H_i).
\end{align}

First consider \eqref{eq:r3-gen-nonsummation}.
Note that by \cref{lem:r3-gen-mu}, in the setting where $p_n = o(1/(n \ln n))$,
we have $\mu = (1+o(1)) n \mu_1$, and hence $1-\frac{\mu}{n \mu_1} = o(1)$.
Furthermore, $p_1^n \le 1$. Thus,
\begin{align*}
    &p_1^n (1 - \frac{\mu}{n \mu_1}) ( n-1 + \frac{1}{p_n} - \frac{\mu}{\mu_1} H_{n-1} - \frac{\mu}{n p_n \mu_1}) \\
    &= o(1) ( n-1 + \frac{1}{p_n} - \frac{\mu}{\mu_1} H_{n-1} - \frac{\mu}{n p_n \mu_1}).
\end{align*}
Note that $\frac{1}{p_n} = \omega(n \ln n)$, so the dominant terms are the two terms involving $\frac{1}{p_n}$. 
\begin{align*}
    ( n-1 + \frac{1}{p_n} - \frac{\mu}{\mu_1} H_{n-1} - \frac{\mu}{n p_n \mu_1}) 
    &=\frac{1 - \frac{\mu}{n \mu_1}}{p_n} + O(n \ln n) = o(1/p_n).
\end{align*}
Note that we only derive an upper bound on the growth rate of \eqref{eq:r3-gen-nonsummation}, namely $o(1/p_n)$. We do not have any lower bound on the growth rate of this term, due to the cancelation of the large quantities $1/p_n$ and $\mu/n p_n \mu_1$. Since these terms are of the same magnitude,  we can only upper bound the result.

Now, consider \eqref{eq:r3-gen-summation}.
Note that $p_1^i = 1+o(1)$ for any $i \le n$ and $-\frac{n}{i} \le 1-\frac{\mu}{i \mu_1} \le 1$.
Specifically, $1-\frac{\mu}{i \mu_1} = O(n/i)$.
Furthermore, $-n \ln i \le i - \frac{\mu}{\mu_1} h_i \le i$.
This gives the bound $i - \frac{\mu}{\mu_1} H_i = O(n \ln n).$

Thus, substituting these bounds into \eqref{eq:r3-gen-summation}, we find that
\begin{align*}
    &\sum_{i=1}^{n-1} p_1^i p_n(1 - \frac{\mu}{i \mu_1}) (i - \frac{\mu}{\mu_1} H_i)
    \le O(p_n) \sum_{i=1}^{n-1} \frac{n^2 \ln n}{i} \\
    =& O(p_n) n^2 \ln nH_{n-1}
    = O(p_n n^2 \ln^2 n)
    = o(n \ln n)
    = o(1/p_n).
    \qedhere
\end{align*}
\end{proof}

\subsection{Regime 3, specific setting:  Polynomial scaling $p_n = 1/n^\alpha, \alpha > 1$}
\label{sec:r3-spec}

In this section, we focus on the specific setting of polynomial $p_n$ scaling where $p_n = 1/n^\alpha, \alpha > 1$.
This allows us to obtain much stronger results. To prove \cref{thm:1-server-dominated}, we start by characterizing the asymptotic behavior of throughput $\mu$ in \cref{lem:r3-spec-mu}:
\begin{lemma}
    \label{lem:r3-spec-mu}
    If $p_n = 1/n^\alpha$ for $\alpha > 1$, then
    \begin{align*}
        \mu = n \mu_1 - \mu_1 n^{2 - \alpha} \ln n
    - (1 + o(1)) \mu_1 (\frac{\mu_1}{\mu_n} - 1) n^{2 - \alpha}.
    \end{align*}
\end{lemma}
\begin{proof}

Consider the general expression for $\mu$ from \cref{lem:explicit-mu}:
\begin{align*}
        \mu &= \frac{1}{p_n} \left( \frac{1}{\mu_n} + \frac{1}{n\mu_1} \frac{p_1^n}{p_n} + \sum_{b=1}^{n-1} \frac{p_1^b}{b \mu_1}\right)^{-1}.
\end{align*}

From \cref{lem:r3-gen-mu}, we know that $\mu = (1+o(1)) n \mu_1$.
Note that for sufficiently large $n$, $\mu_n < n \mu_1$, and as a result $\mu < n \mu_1$, because the completion rate $\mu_i$ in state $i$ is at most $n \mu_1$ for all states $i$, for all $n \ge \frac{\mu_n}{\mu_1}$.

We therefore examine $\frac{n \mu_1}{\mu} - 1$:
\begin{align}
    \nonumber
    \frac{n \mu_1}{\mu} - 1 &= \left(n \mu_1 p_n \left( \frac{1}{\mu_n} + \frac{1}{n\mu_1} \frac{p_1^n}{p_n} + \sum_{b=1}^{n-1} \frac{p_1^b}{b \mu_1}\right) - 1\right) \\
    \nonumber
    &= \frac{n \mu_1 p_n}{\mu_n} + \frac{n \mu_1 p_n}{n\mu_1} \frac{p_1^n}{p_n} + n \mu_1 p_n \sum_{b=1}^{n-1} \frac{p_1^b}{b \mu_1} - 1 \\
    \label{eq:r3-spec-mu-first}
    &= \frac{n^{1-\alpha} \mu_1}{\mu_n} + p_1^n - 1 + n^{1-\alpha} \sum_{b=1}^{n-1} \frac{p_1^b}{b}.
\end{align}

We next expand $p_1^n = (1- n^{-\alpha})^n$ as
\begin{align*}
    (1- n^{-\alpha})^n = 1 - n n^{-\alpha} + \frac{n(n-1)}{2} n^{-2\alpha} \ldots
\end{align*}

Because $\alpha > 1$, we can summarize this expansion as:
\begin{align*}
    p_1^n = 1 - n^{1-\alpha} + O(n^{2-2\alpha}).
\end{align*}

Substituting the asymptotic behavior of $p_1^n$ into \eqref{eq:r3-spec-mu-first}, we find that
\begin{align*}
    \frac{n^{1-\alpha} \mu_1}{\mu_n} + p_1^n - 1 + n^{1-\alpha} \sum_{b=1}^{n-1} \frac{p_1^b}{b}
    = \frac{n^{1-\alpha} \mu_1}{\mu_n} - n^{1-\alpha} + O(n^{2-2\alpha})
    + n^{1-\alpha} \sum_{b=1}^{n-1} \frac{p_1^b}{b}.
\end{align*}

Note also that $p_1^b \in [p_1^n, 1]$ for all $b \le n$.
In particular, $p_1^b = 1+ O(n^{1-\alpha})$ (uniformly) for all $b \le n$.

We can therefore simplify the summation as

\begin{align*}
    n^{1-\alpha} \sum_{b=1}^{n-1} \frac{p_1^b}{b} = n^{1-\alpha} (1 + O(n^{1-\alpha})) \sum_{b=1}^{n-1} \frac{1}{b} = n^{1-\alpha}\ln n + O(n^{2-2\alpha} \ln n)).
\end{align*}

Putting it together, we find that
\begin{align}
    \nonumber
    \frac{n \mu_1}{\mu} - 1 &= \frac{n^{1-\alpha} \mu_1}{\mu_n} - n^{1-\alpha} + O(n^{2-2\alpha})
    + n^{1-\alpha}\ln n + O(n^{2-2\alpha})) \\
    \label{eq:r3-spec-mu-2}
    &= n^{1-\alpha} \left(\ln n + \frac{\mu_1}{\mu_n} - 1\right) +  O(n^{2-2\alpha} \ln n).
\end{align}

Let $x$ denote the growth rate of the final expression \eqref{eq:r3-spec-mu-2}, and use it to solve for $\mu$:

\begin{align*}
    \frac{n \mu_1}{\mu} - 1 &= x \implies
    n \mu_1 - \mu = x \mu \implies
    n \mu_1 = (1 + x) \mu \implies
    \frac{n \mu_1}{1 + x} = \mu.
\end{align*}

Note that for any $x'$ near zero, $\frac{1}{1+x'} = 1 - x' + O(x'^2)$. Furthermore, since $\alpha > 1$, $x$ is in fact near 0. Thus, we have

\begin{align*}
    \mu = \frac{n \mu_1}{1 + x} &= n \mu_1 (1 - x + O(x^2)) \\
    &=n \mu_1 - n \mu_1  \left(n^{1-\alpha} \left(\ln n + \frac{\mu_1}{\mu_n} - 1\right) +  O(n^{2-2\alpha} \ln n) + O(n^{2-2\alpha} \ln^2 n) \right) \\
    &= n \mu_1 - \mu_1 n^{2-\alpha} \ln n - \mu_1\left(\frac{\mu_1}{\mu_n} - 1\right) n^{2 - \alpha} + O(n^{3-2\alpha}\ln^2 n).
\end{align*}

Since $\alpha > 1$, we have $3 - 2\alpha < 2 - \alpha$, so our final expression matches our theorem statement with a more specific error bound.
\end{proof}

We are now ready to prove the remainder of \cref{thm:1-server-dominated} by characterizing the asymptotic mean queue length $\E[\Delta(Y_d)]$.

\begin{theorem}
    \label{thm:r3-spec-delta}
    If $p_n = 1/n^\alpha$ for $\alpha > 1$, then
    \begin{align*}
        \E[\Delta(Y_d)] = (1+o(1))\frac{1}{2} n^{2 - \alpha} \ln^2 n.
    \end{align*}
\end{theorem}
\begin{proof}

We start by splitting up $\E[\Delta(Y_d)]$ into five terms, as before:
\begin{align}
    \E[\Delta(Y_d)] &=
    \label{eq:r3-non-summation}
    p_1^n (1 - \frac{\mu}{n \mu_1}) ( n-1 + \frac{1}{p_n} - \frac{\mu}{\mu_1} H_{n-1} - \frac{\mu}{n p_n \mu_1}) \\
    \label{eq:r3-p1i-i}
    + &p_n \sum_{i=1}^{n-1} p_1^i i \\
    \label{eq:r3-p1i}
    - &\frac{\mu p_n}{\mu_1} \sum_{i=1}^{n-1} p_1^i \\
    \label{eq:r3-p1i-hi}
    - &\frac{\mu p_n}{\mu_1} \sum_{i=1}^{n-1} p_1^i H_i \\
    \label{eq:r3-p1i-hi-i}
    + &\frac{\mu^2 p_n}{\mu_1^2} \sum_{i=1}^{n-1} p_1^i \frac{1}{i} H_i
\end{align}
We will demonstrate the following asymptotic growth rates for each of these terms:
\begin{enumerate}
    \item Non-summation term \eqref{eq:r3-non-summation}: $(1+o(1)) (\frac{\mu_1}{\mu_n} - \gamma) n^{2-\alpha} \ln n$.
    \label{it:r3-1}
    \item $p_1^i i$ term \eqref{eq:r3-p1i-i}: $O(n^{2-\alpha})$.
    \label{it:r3-2}
    \item $p_1^i$ term \eqref{eq:r3-p1i}: $- (1+o(1)) p_1  n^{2 - \alpha}$.
    \label{it:r3-3}
    \item $p_1^i H_i$ term \eqref{eq:r3-p1i-hi}: $-n^{2-\alpha} \ln n + O(n^{2-\alpha})$.
    \label{it:r3-4}
    \item $p_1^i H_i / i$ term \eqref{eq:r3-p1i-hi-i}: $\frac{1}{2}n^{2-\alpha}\ln^2 n + O(n^{2-\alpha}\ln n)$.
    \label{it:r3-5}
\end{enumerate}
Note that only \cref{it:r3-1}, the asymptotic growth rate for \eqref{eq:r3-p1i-hi-i}, is non-negligible. As a result, these growth rates sum to $\frac{1}{2}n^{2-\alpha}\ln^2 + O(n^{2-\alpha} \ln n)$, as desired.

We start with the non-summation term \eqref{eq:r3-non-summation}:
\begin{align*}
    p_1^n (1 - \frac{\mu}{n \mu_1}) ( n-1 + \frac{1}{p_n} - \frac{\mu}{\mu_1} H_{n-1} - \frac{\mu}{n p_n \mu_1}).
\end{align*}

We will consider the three multiplicative terms one by one, left to right.
First, note that $p_1^n = 1+o(1)$.
Then, we expand $1-\frac{\mu}{n \mu_1}$. Using \cref{lem:r3-spec-mu}:
\begin{align*}
    1 - \frac{\mu}{n \mu_1} = \frac{n \mu_1 - \mu}{n \mu_1} = n^{1-\alpha}\ln n + (1+o(1)) (\frac{\mu_1}{\mu_n} -1) n^{1-\alpha}.
\end{align*}

Next expand the final (rightmost) parenthesis of \eqref{eq:r3-non-summation}. Again using \cref{lem:r3-spec-mu}, we have:
\begin{align}
    \nonumber
    &n-1 + \frac{1}{p_n} - \frac{\mu}{\mu_1} H_{n-1} - \frac{\mu}{n p_n \mu_1} 
    = n-1 - \frac{\mu}{\mu_1} H_{n-1} + n^\alpha (1-\frac{\mu}{n \mu_1}) \\
    \nonumber
    =& n-1 - \frac{\mu}{\mu_1} H_{n-1} + n \ln n + (1+o(1)) n (\frac{\mu_1}{\mu_n} -1)\\
    \label{eq:rightmost}
    =& n-1 - (n - \frac{\ln n}{n^{2-\alpha }}
    - (1 + o(1)) (\frac{\mu_1}{\mu_n} - 1)) n^{2 - \alpha}) H_{n-1} + n \ln n + (1+o(1)) n (\frac{\mu_1}{\mu_n} -1).
\end{align}
Note that  $2 - \alpha < 1$. Let $\epsilon > 0$ be a constant such that $2 - \alpha< 1-\epsilon < 1$.
As a result, we can find the asymptotic growth rate of \eqref{eq:rightmost} as follows:
\begin{align*}
    & n-1 - (n - n^{2 - \alpha} \ln n
    - (1 + o(1)) (\frac{\mu_1}{\mu_n} - 1)) n^{2 - \alpha}) H_{n-1} + n \ln n + (1+o(1)) n (\frac{\mu_1}{\mu_n} -1)\\
    =& n-1 - (n - O(n^{1-\epsilon})) H_{n-1} + n \ln n + (1+o(1)) n (\frac{\mu_1}{\mu_n} -1)\\
    =& n-1 - (n - O(n^{1-\epsilon})) (\ln n + \gamma + O(1/n)) + n \ln n + (1+o(1)) n (\frac{\mu_1}{\mu_n} -1))\\
    =& n-1 + O(n^{1-\epsilon}) (\ln n + \gamma + O(1/n)) - n (\gamma + O(1/n))+ (1+o(1)) n (\frac{\mu_1}{\mu_n} -1)\\
    =& n(\frac{\mu_1}{\mu_n} - \gamma) +o(n).
\end{align*}

Combining these results, we find that the overall growth rate of $\eqref{eq:r3-non-summation}$ is
\begin{align*}
    (1+o(1)) (\frac{\mu_1}{\mu_n} - \gamma) n^{2-\alpha} \ln n,
\end{align*}
confirming \cref{it:r3-1}.

We next consider \eqref{eq:r3-p1i-i} which can be quantified exactly as:
\begin{align*}
    \sum_{i=1}^{n-1} p_1^i p_n i &= p_n \frac{p_1 - n p_1^n + (n-1) p_1^{1 + n}}{(1 - p_1)^2} = \frac{p_1 - n p_1^n +n p_1 p_1^n - p_1 p_1^n}{p_n} \\
    &= \frac{p_1 ( 1-p_1^n) - n p_1^n (1-p_1)}{p_n} = \frac{p_1 ( 1-p_1^n) - n p_1^n (1-p_1)}{p_n}.
\end{align*}
We now expand $p_1^n = (1-p_n)^n$ using a binomial expansion, and recalling that $n p_n = o(1)$,
\begin{align*}
    &\frac{p_1 ( 1-p_1^n) - n p_1^n (1-p_1)}{p_n}
    = \frac{p_1 (n p_n + O(n^2 p_n^2)) - n p_n (1 - n p_n + O(n^2 p_n^2))}{p_n} \\
    &= \frac{-(1-p_1) n p_n + O(n^2 p_n^2)}{p_n}
    = \frac{-n p_n^2 + O(n^2 p_n^2)}{p_n}
    = \frac{O(n^2 p_n^2)}{p_n}
    = O(n^2 p_n)
    = O(n^{2 - \alpha}),
\end{align*}
confirming \cref{it:r3-2}.

Next consider \eqref{eq:r3-p1i} which again has an exact expression
\begin{align*}
    -\sum_{i=1}^{n-1} p_1^i p_n \frac{\mu}{\mu_1}
    &= - \frac{\mu}{\mu_1} p_n \frac{p_1-p_1^n}{1-p_1} = -\frac{\mu}{\mu_1}(p_1 - p_1^n) = -\frac{\mu}{\mu_1} p_1(1-p_1^{n-1}) \\
    &= - \frac{\mu}{\mu_1}p_1 (1+o(1)) n^{1-\alpha} = - p_1 (1+o(1)) n^{2 - \alpha},
\end{align*}
confirming \cref{it:r3-3}.

We next consider \eqref{eq:r3-p1i-hi}:
\begin{align}
    \nonumber
    -\sum_{i=1}^{n-1} p_1^i p_n \frac{\mu}{\mu_1} H_i
    &= -\sum_{i=1}^{n-1} p_1^i p_n \frac{\mu}{\mu_1} \ln i + O(n^{2-\alpha})\\
    \nonumber
    &= -p_n \frac{\mu}{\mu_1} \int_{i=1}^n p_1^i  \ln i di + O(n^{2-\alpha}) \\
    \label{eq:r3-p1i-hi-first}
    &= -p_n \frac{\mu}{\mu_1} \int_{i=1}^n (1-n^{-\alpha})^i  \ln i di + O(n^{2-\alpha}).
\end{align}

As usual, we convert the $H_i$ to $\ln i$ and the summation to an integral.
To bound the change caused by switching from $H_i$ to $\ln i$, note that $|H_i - \ln i| \le 1$, which gives rise to an error bounded by \eqref{eq:r3-p1i}, which is negligible compared to the growth rate we will establish for this term.

The change caused by switching from the summation to the integral is similarly bounded using the fact that $\ln i - \ln (i-1) \le 1$  (which again introduces an error bounded by \eqref{eq:r3-p1i}).

Now, we upper and lower bound the integral as follows:
\begin{align*}
     \int_{i=1}^n p_1^i  \ln i di &\le \int_{i=1}^n \ln i di = n \ln n - n +1,\\
     \int_{i=1}^n p_1^i \ln i di &\ge \int_{i=1}^n p_1^n \ln i di = (n \ln n - n +1) p_1^n.
\end{align*}
Using a binomial expansion of $p_1^n$, we have
\begin{align*}
    p_1^n = (1 - n^{-\alpha})^n
    &= 1 - n n^{-\alpha} + \frac{n(n-1)}{2} n^{-2\alpha} - \ldots = 1 - O(n^{1-\alpha}).
\end{align*}

We can therefore analyze the bound above as:
\begin{align*}
    (n \ln n - n +1) p_1^n &= (n \ln n - n +1) (1- O(n^{1-\alpha}))  \\
    &= n \ln n - n - O(n^{2-\alpha} \ln n) = n \ln n - O(n) \\
    \implies \int_{i=1}^n p_1^i \ln i di &= n \ln n - O(n).
\end{align*}

Substituting the integral into \eqref{eq:r3-p1i-hi-first}, multiplying by $-p_n \mu/\mu_1 = n^{1-\alpha} + O(n^{2-2\alpha} \ln n)$, and 
using \cref{lem:r3-spec-mu}, we get
\begin{align*}
    -\sum_{i=1}^{n-1} p_1^i p_n \frac{\mu}{\mu_1} H_i = -n^{2-\alpha} \ln n + O(n^{2-\alpha}),
\end{align*}
which matches \cref{it:r3-4}, the asymptotic growth rate for \eqref{eq:r3-p1i-hi}.

Finally, consider \eqref{eq:r3-p1i-hi-i}. As usual, we convert the $H_i$ to $\ln i$ and  the summation to an integral:
\begin{align}
    \label{eq:r3-p1i-hi-i-first}
    \sum_{i=1}^{n-1} p_1^i p_n \frac{\mu^2}{i \mu_1^2} H_i = p_n \frac{\mu^2}{\mu_1^2} \int_{i=1}^n p_1^i \frac{\ln i}{i} di + O(n^{2-\alpha} \ln n).
\end{align}
To justify the error term, note that both conversions introduce an error in the order of
\begin{align*}
    p_n \frac{\mu^2}{\mu_1} \sum_{i=1}^{n-1} p_1^i \frac{1}{i}.
\end{align*}

Making use of our bound from \cref{lem:r3-spec-mu}, which states that $\mu = n \mu_1 + O(n^{1-\alpha} \ln)$,
and the earlier shown result $p_1^i = 1 + O(n^{1-\alpha})$,
we have
\begin{align*}
    p_n \frac{\mu^2}{\mu_1} \sum_{i=1}^{n-1} p_1^i \frac{1}{i} = O(n^{2-\alpha} \ln n),
\end{align*}
which is the error term we used in \eqref{eq:r3-p1i-hi-i-first}.

We upper and lower bound the integral from \eqref{eq:r3-p1i-hi-i-first} as follows:
\begin{align}
    \label{eq:r3-p1i-hi-i-lower}
    \int_{i=1}^n p_1^i \frac{\ln i}{i} di &\le \int_{i=1}^n \frac{\ln i}{i} di = \frac{\ln^2 n}{2}, \\
    \nonumber
    \int_{i=1}^n p_1^i \frac{\ln i}{i} di &\ge \int_{i=1}^n p_1^n \frac{\ln i}{i} di = p_1^n
    \frac{\ln^2 n}{2}.
\end{align}

Using our prior result $p_1^n = 1 + O(n^{1-\alpha})$,
we get
\begin{align*}
    \int_{i=1}^n p_1^i \frac{\ln i}{i} di &\ge ( 1 + O(n^{1-\alpha})) \frac{\ln^2 n}{2}
    = \frac{\ln^2 n}{2} + O(n^{1-\alpha}\frac{\ln^2 n}{2}).
\end{align*}
Combining this upper bound with our lower bound in \eqref{eq:r3-p1i-hi-i-lower}, we obtain
\begin{align}    
    \label{eq:r3-p1i-hi-i-middle}
    \int_{=1}^n p_1^i \frac{\ln i}{i} di &= \frac{\ln^2 n}{2} + O(n^{1-\alpha}\frac{\ln^2 n}{2}).
\end{align}

We now need to incorporate the $p_n \frac{\mu^2}{\mu_1^2}$ term.
From \cref{lem:r3-spec-mu}, we have $\frac{\mu}{\mu_1} = n + O(n^{2-\alpha} \ln n)$. Then
\begin{align}
    \label{eq:r3-p1i-hi-i-last}
    p_n \frac{\mu^2}{\mu_1^2} = n^{-\alpha} (n + O(n^{2-\alpha} \ln n))^2
    = n^{-\alpha}(n^2 + O(n^{3-\alpha} \ln n))
    = n^{2-\alpha} + O(n^{3 - 2 \alpha} \ln n).
\end{align}

We now combine \eqref{eq:r3-p1i-hi-i-middle} and \eqref{eq:r3-p1i-hi-i-last} to get our result:
\begin{align*}
    \sum_{i=1}^{n-1} p_1^i p_n \frac{\mu^2}{i \mu_1^2} H_i &=
    \frac{1}{2} n^{2-\alpha} \ln^2 n + O(n^{2-\alpha} \ln n),
\end{align*}
which matches the claimed asymptotic growth rate for \eqref{eq:r3-p1i-hi-i}, and completing the proof.
\end{proof}

\section{Comparison: Asymptotic growth rate versus fixed-$n$ formula}
\label{sec:empirical}

We have asymptotically characterized throughput $\mu$ and scaled mean queue length $\E[\Delta(Y_d)]$
for each of our three regimes of interest, namely, $n$-server dominated, balanced load, and 1-server dominated, \textcolor{black}{in \cref{thm:1-server-dominated,thm:n-server-dominated}, within the load-focused multiserver scaling limit}. In this section, we numerically evaluate our asymptotic results, by comparing our asymptotic formulas with the explicit summation-based formulas in \cref{lem:explicit-mu,lem:explicit-delta}. \textcolor{black}{In \cref{extensions}, we derive practical insights from our asymptotic expressions and demonstrate that those insights generalize to settings beyond the ones studied in this paper, via numerical evaluation and simulation.}

We numerically evaluate these formulas at numbers of servers up to $n=10^8$.
We choose this upper limit because modern exascale datacenters reach about this number of CPU cores.
For instance, the Citadel Campus \cite{citadel} reports 650 MW of planned total energy consumption, corresponding to about 1-2 million 1U server racks and about 100 million CPU cores.
Values of $n$ around $10^8$ thus represent the approximate scale of the
largest multiserver-job queueing systems currently in operation or construction.

We explore three parameter families, setting $p_n = n^{-\alpha}$ for $\alpha \in \{0.5, 1, 2\}$.
These three values of $\alpha$ correspond to our three regimes.
In \cref{fig:n-dominated}, we compare the exact values of $\mu$ and $\E[\Delta(Y_d)]$ with the asymptotic growth rates in the $\alpha = 0.5$ parametrization, which is in the $n$-server dominated regime.
In \cref{fig:balanced}, we examine the $\alpha = 1$ parametrization, which is in the balanced load regime.
In \cref{fig:1-dominated}, we examine the $\alpha = 2$ parametrization, which is in the 1-server dominated regime. 
In this section, we specifically focus on the $\mu_1 = \mu_n = 1$ parameterization. See \cref{app:additional-numerical} for further parameterizations of $\mu_1$ and $\mu_n$.

\begin{figure}
    \centering
    \begin{subfigure}{0.49\linewidth}
        \includegraphics[width=\textwidth]{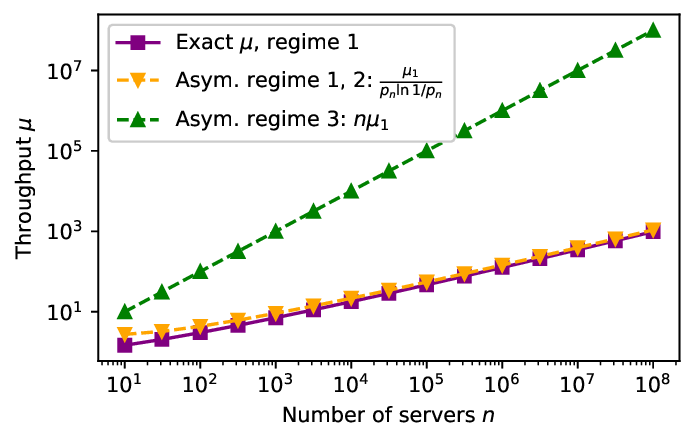}
    \end{subfigure}
    \begin{subfigure}{0.49\linewidth}
        \includegraphics[width=\textwidth]{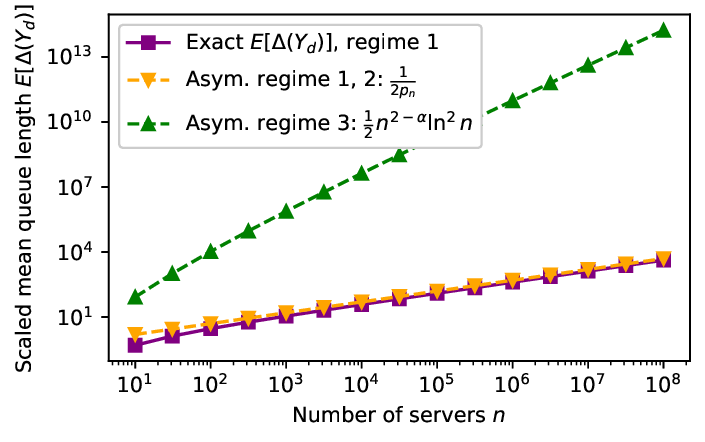}
    \end{subfigure}
    \caption{Exact versus asymptotic formulas in the $n$-server dominated load regime for $\mu$ and  $\E[\Delta(Y_d)]$ as functions of  $n$. Parametrization: $p_n=n^{-0.5}$ and $\mu_1=\mu_n=1$.}
    \label{fig:n-dominated}
\end{figure}

\begin{figure}
    \centering
    \begin{subfigure}{0.49\linewidth}
        \includegraphics[width=\textwidth]{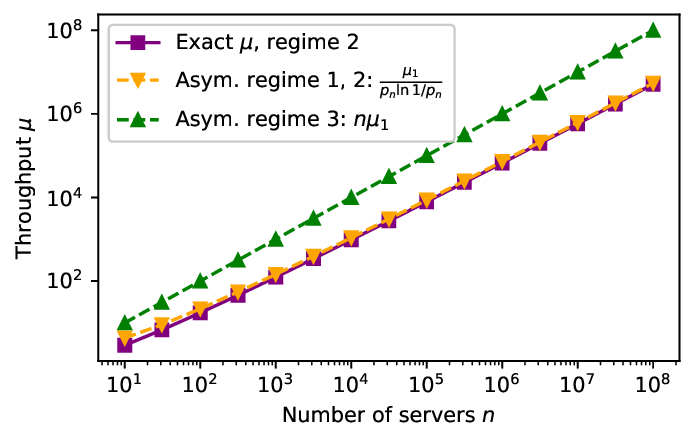}
    \end{subfigure}
    \begin{subfigure}{0.49\linewidth}
        \includegraphics[width=\textwidth]{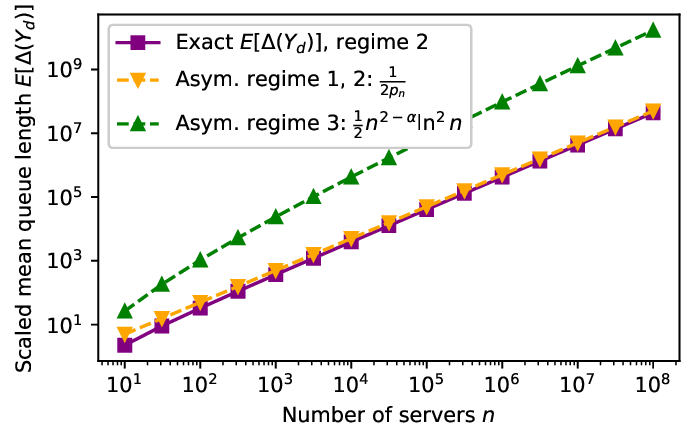}
    \end{subfigure}
    \caption{Exact versus asymptotic formulas in the balanced load regime for $\mu$ and  $\E[\Delta(Y_d)]$ as functions of  $n$. Parametrization: $p_n=n^{-1}$ and $\mu_1=\mu_n=1$.}
    \label{fig:balanced}
\end{figure}

\begin{figure}
    \centering
    \begin{subfigure}{0.49\linewidth}
        \includegraphics[width=\textwidth]{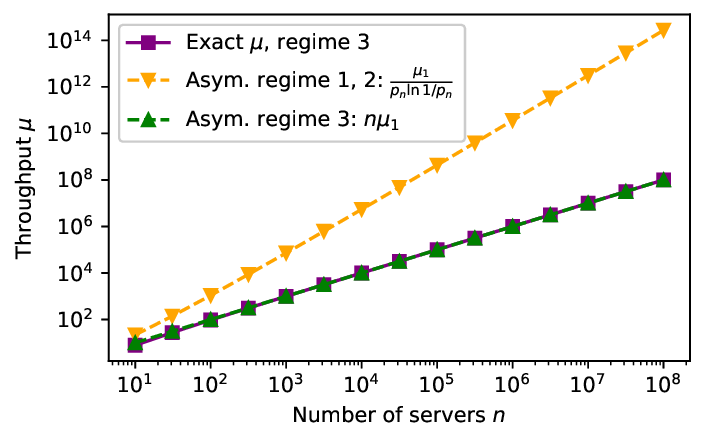}
    \end{subfigure}
    \begin{subfigure}{0.49\linewidth}
        \includegraphics[width=\textwidth]{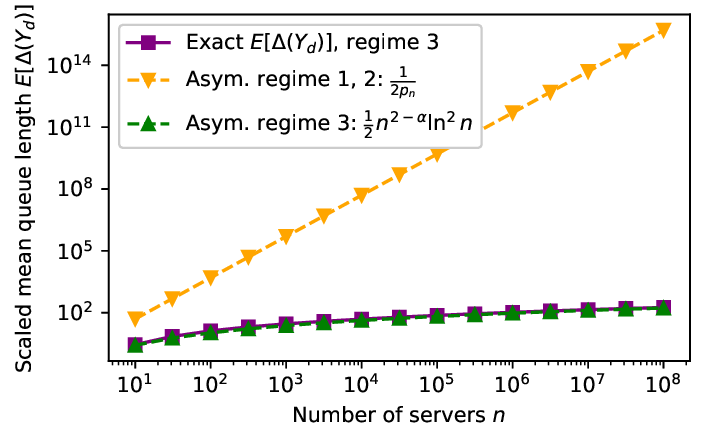}
    \end{subfigure}
    \caption{Exact versus asymptotic formulas in the $1$-server dominated regime for $\mu$ and  $\E[\Delta(Y_d)]$ as functions of  $n$. Parametrization: $p_n=n^{-2}$ and $\mu_1=\mu_n=1$.}
    \label{fig:1-dominated}
\end{figure}

For each parameter family, we compare the exact formulas against each of our two asymptotic growth rates. In each case, the exact values converge to one of the two asymptotic growth rates in the $n \to \infty$ limit, as predicted by our theoretical results.
Specifically, in \cref{fig:n-dominated}, where we examine the $n$-server dominated regime,
and in \cref{fig:balanced}, where we examine the balanced regime,
we compare the exact value of $\mu$, given in purple, with the asymptotic formula $\frac{\mu_1}{p_n \ln 1/p_n}$, in orange, and exact value of $\E[\Delta(Y_d)]$ with the asymptotic formula $1/2p_n$.
\textcolor{black}{In \cref{thm:n-server-dominated}} we proved that these two formulas, the purple and orange curves, differ by a ratio of $1+o(1)$
as $n \to \infty, p_n \to 0$.
We plot these formulas on a log-log plot, so this multiplicative convergence corresponds to a converging vertical gap.

Throughout the paper, for the benefit of monochromatic readers, the green curves use upward-triangle markers, the orange curves use downward-triangle markers, and the purple curves use square markers.

The convergence is clearly apparent over the interval $n \in [10^1, 10^8]$ spanned by our plots, confirming our theoretical results in both settings. For comparison, we also plot the formulas $n \mu_1$ against $\mu$ and $\frac{1}{2} n^{2 - \alpha} \ln^2 n$ against $\E[\Delta(Y_d)]$, in green, which are the formulas that we derived for the 1-server dominated regime. As expected, this green curve diverges from the true value and from the correct asymptotic curve.

In \cref{fig:1-dominated}, we examine the $1$-server dominated regime, and as expected from \cref{thm:n-server-dominated}, the tight asymptotics are the green curves in this setting: $n \mu_1$ for $\mu$ and $\frac{1}{2} n^{2 - \alpha} \ln^2 n$ for $\E[\Delta(Y_d)]$.
In fact, throughout the range of $n$ that we examine, there is nearly no perceptible separation between the purple and green curves, indicating not just multiplicative convergence, but a tight multiplicative approximation across all $n$ and $p_n$.

\section{Empirical Evaluation of Extension Models} \label{extensions}

\textcolor{black}{Our asymptotic theoretical results imply two practical insights in the 1-and-$n$ model, as the mix of jobs changes from $n$-server focused to 1-server focused:}
\begin{enumerate}
    \item \textcolor{black}{The throughput of the system rises monotonically, with an inflection point around regime 2, where the loads of both types of jobs are similar.}
    \item \textcolor{black}{The mean response time of the system, holding fraction-of-throughput constant, is a bell curve, rising for $n$-server-dominated systems, peaking for a balanced mixture, and falling again for $1$-server dominated systems.}
\end{enumerate}

\textcolor{black}{We will show that these insights hold not just in the 1-and-$n$ system that is the primary focus of this paper, but also in the following three generalizations:}
\begin{itemize}
    \item \textcolor{black}{A system with an additional scaling behavior: 1-server jobs complete at rate $\mu_1 = n$, while $\mu_n = 1$, so the $n$-server jobs have much longer duration than the 1-server jobs, in addition to their larger server need. (See \cref{sec:emp-duration-scaling})}
    \item \textcolor{black}{A system where the large-resource jobs require $n/2$ servers, rather than $n$ servers. (See \cref{sec:emp-half})}
    \item \textcolor{black}{A system with three classes of jobs: $n$-server, $n/2$-server and $1$-server jobs. (See \cref{sec:emp-three})}
\end{itemize}

\textcolor{black}{To investigate these systems, we use a combination of exact calculations, using the formulas proven in  \cref{lem:explicit-mu,lem:explicit-delta} where applicable,
as well as simulation. In all cases, we show that the same two practical insights apply, appropriately tailored to the relevant systems.}

\textcolor{black}{We start in \cref{sec:emp-original} by investigating the original 1-and-$n$ system via exact calculation and simulation, to obtain a baseline for the behavior corresponding to our insights, then turn to the extensions.}

\subsection{Calculations and Simulation for 1-and-$n$ System}
\label{sec:emp-original}

\begin{figure}
    \centering
    \begin{subfigure}[t][][t]{0.49\linewidth}
        \includegraphics[width=\textwidth]{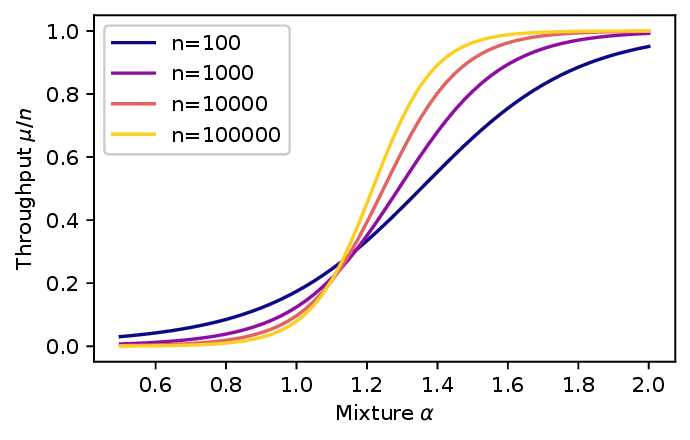}
        \caption{Throughput $\mu/n$, calculated using \cref{lem:explicit-mu}.}
        \label{fig:original-calc-mu}
    \end{subfigure}
    \begin{subfigure}[t][][t]{0.49\linewidth}
        \includegraphics[width=\textwidth]{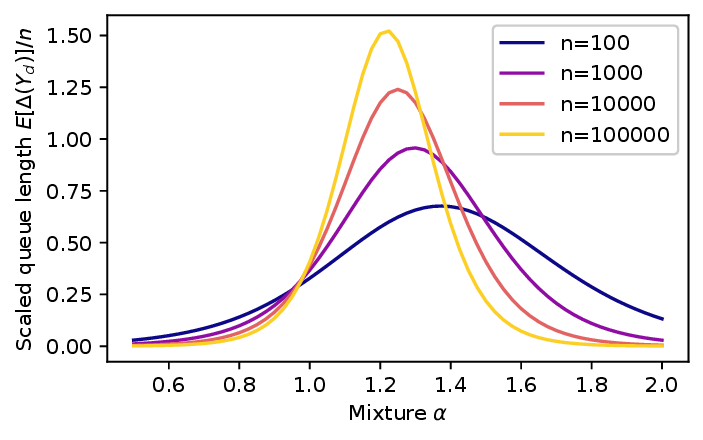}
        \caption{Scaled mean queue length $\E[\Delta(Y_d)]/n$, calculated using \cref{lem:explicit-delta}.}
        \label{fig:original-calc-delta}
    \end{subfigure}
    \caption{Main setting with server needs $1$ and $n$, $\mu_1 = \mu_n = 1$, and  $p_n = 1/n^\alpha$.}
    \label{fig:original-calc}
\end{figure}

\begin{figure}
    \centering
    \begin{subfigure}[t][][t]{0.49\linewidth}
        \includegraphics[width=\textwidth]{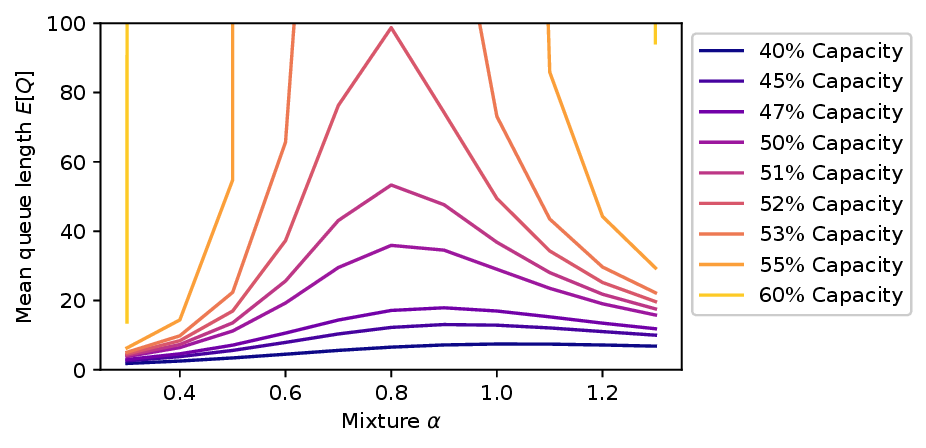}
        \caption{Main setting with  server needs $1$ and  10, $n=10$ servers and $\mu_1 = \mu_{10} = 1$.}
        \label{fig:original-sim}
    \end{subfigure}
    \begin{subfigure}[t][][t]{0.49\linewidth}
        \includegraphics[width=\textwidth]{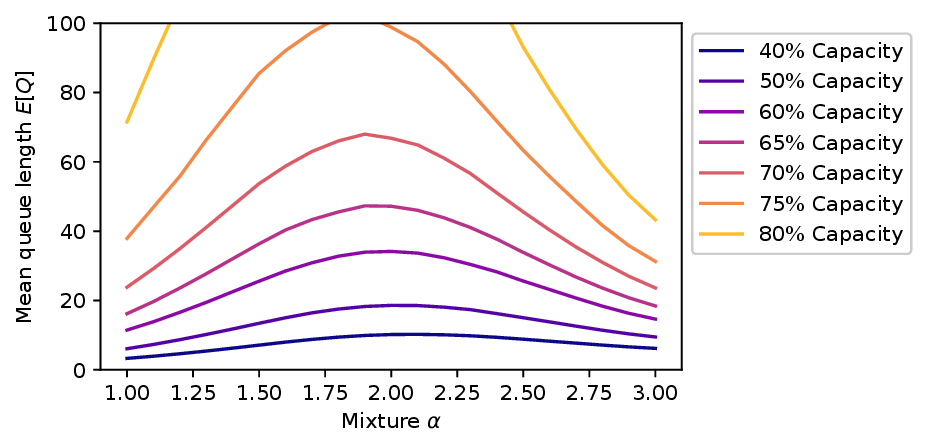}
        \caption{Equal-area setting with server needs $1$ and $10$, $n=10$ servers and $\mu_1 = 1, \mu_{10} = 10$.}
        \label{fig:area-sim}
    \end{subfigure}
    \caption{Simulation results ($10^8$ jobs per data point) for mean queue length when $p_{10} = 1/n^\alpha$  and  $\lambda$ is chosen so that a constant fraction of capacity is utilized.}
    \label{fig:original-and-area-sim}
\end{figure}

\textcolor{black}{In \cref{fig:original-calc}, we see both insights verified in the 1-and-$n$ system,
via direct calculation of throughput $\mu$ and the scaled queue length $\E[\Delta(Y_d)]$
using \cref{lem:explicit-mu,lem:explicit-delta}.
Each curve represents a fixed number of servers $n$ and varying parameter $\alpha$ controlling the fraction of large jobs $p_n = \frac{1}{n^\alpha}$. We scale throughput $\mu$ and scaled queue length $\E[\Delta(Y_d)]$ by a factor of $n$ to make the curves comparable.}

\textcolor{black}{In \cref{fig:original-calc-mu}, we observe that throughput increases monotonically with $\alpha$, with an inflection point around $\alpha \approx 1.2$, as predicted by \cref{thm:n-server-dominated,thm:1-server-dominated}.}

\textcolor{black}{In \cref{fig:original-calc-delta}, we observe that scaled mean queue length increases with $\alpha$ to a peak around $\alpha=1.2$ and  then decreases back to 0 as $\alpha$ gets large, also as predicted by \cref{thm:n-server-dominated,thm:1-server-dominated}.}

\textcolor{black}{In \cref{fig:original-sim}, we observe via simulation that these insights are not restricted to calculations or to very large $n$. Here, we simulate $n=10$ servers, with each curve representing a different fraction of capacity demanded, using the following formula:}
\begin{align*}
\textcolor{black}{\text{Fraction of capacity} := \frac{\mu_1 p_1 + \mu_n p_n n}{n}}    
\end{align*}

\textcolor{black}{The bell-curve behavior of mean queue length is clearly visible, peaking around $\alpha=0.8$, and is especially clear for capacity fractions around the boundary of the stability region.
Our bell-curve insight thus applies not only for scaled mean queue length $\E[\Delta(Y_d)]$,
but also for mean queue length $\E[Q]$ itself, including in systems with as few as $n=10$ servers.}

\subsection{Calculations and Simulation for Duration-Scaling System}
\label{sec:emp-duration-scaling}
\begin{figure}
    \centering
    \begin{subfigure}[t][][t]{0.49\linewidth}
        \includegraphics[width=\textwidth]{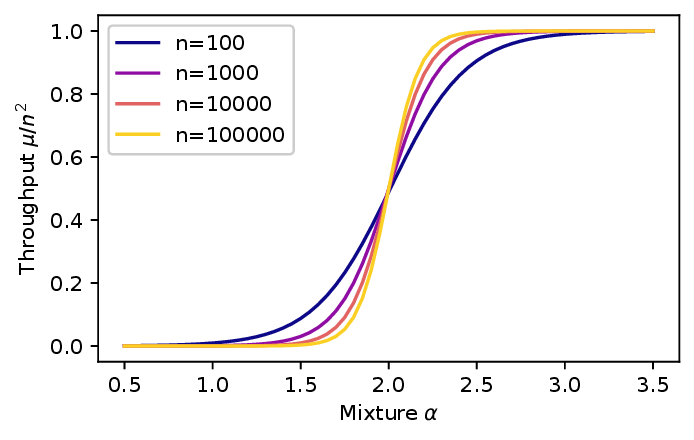}
        \caption{Throughput $\mu/n^2$, calculated using \cref{lem:explicit-mu}.}
    \label{fig:area-calc-mu}
    \end{subfigure}
    \begin{subfigure}[t][][t]{0.49\linewidth}
        \includegraphics[width=\textwidth]{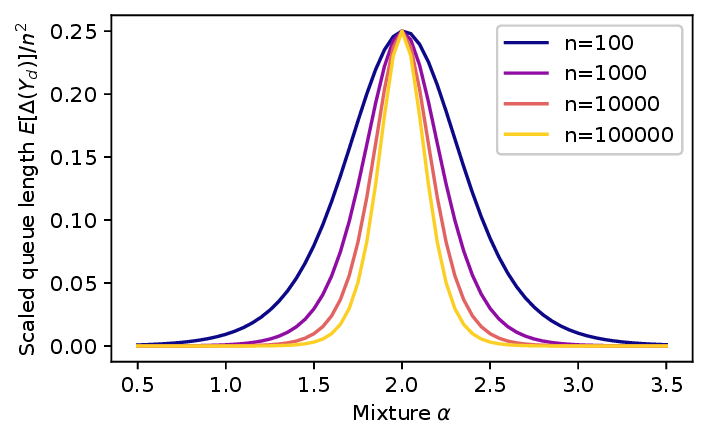}
        \caption{Scaled mean queue length $\E[\Delta(Y_d)]/n^2$, calculated using \cref{lem:explicit-delta}.}
    \label{fig:area-calc-delta}
    \end{subfigure}    
    \caption{ Scaled service rate setting with server needs $1$ and $n$, $\mu_1 = n, \mu_n = 1$ and $p_n = 1/n^\alpha$.}
    \label{fig:area-calc}
\end{figure}

\textcolor{black}{In this section, we investigate the realistic scenario where low-resource jobs also require much lower service duration than high-resource jobs. This is realistic in datacenters, for instance, where a single-server analytics job might only take a few seconds, while an LLM-training job might take most of the datacenter for a period of months. To model this scenario, we set the service rate $\mu_1$ of the 1-server jobs equal to $n$, the number of servers, while leaving the server rate $\mu_n$ of the $n$-server jobs constant.}

\textcolor{black}{In \cref{fig:area-calc}, we again observe that  both of our insights are verified even more clearly in the duration-scaling system where throughput $\mu$ and the scaled queue length $\E[\Delta(Y_d)]$ are computed 
using \cref{lem:explicit-mu,lem:explicit-delta}. Note that we use the scaling $\mu/n^2$ and $\E[\Delta(Y_d)]/n^2$
to account for the fact that the 1-server jobs have both $1/n$ the mean duration  and $1/n$ the resource demand of the $n$-server jobs.}

\textcolor{black}{\cref{fig:area-calc-mu} shows that throughput $\mu$ has an inflection point at $\alpha=2$, and \cref{fig:area-calc-delta} shows that scaled mean queue length has a particularly symmetrical bell curve centered on $\alpha=2$ as well. We expect the inflection point to be around $\alpha=2$, because this is where the load of $n$-server and $1$-server jobs is equalized. Thus our insights are verified in this setting as well, despite leaving the domain of applicability of asymptotic results.}

\textcolor{black}{Note that \cref{fig:area-sim} shows that our insights again carry over into simulation and into low-$n$ settings, with mean queue length $\E[Q]$ peaking around $\alpha =1.9$, across a wide range of fractions of capacity for $n=10$ servers. Hence, the bell-curve behavior is even smoother and more consistent in this duration-scaling setting than in \cref{sec:emp-original}.}

\subsection{Calculations and Simulation for Half-Size Large Jobs}
\label{sec:emp-half}

\begin{figure}
    \centering
    \begin{subfigure}[t][][t]{0.42\linewidth}
        \includegraphics[width=\textwidth]{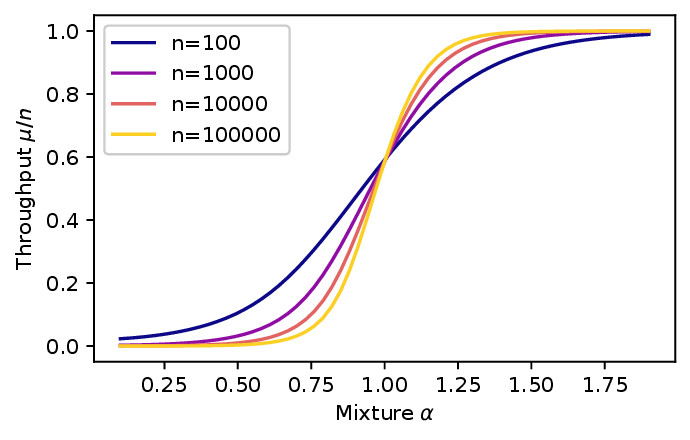}
    \caption{Throughput $\mu/n$, calculated using \cite[Section 5]{grosof_new_2023}, \cite[Section 5.1]{comte_graph_2025}.}
    \label{fig:half-calc-mu}
    \end{subfigure}
    \begin{subfigure}[t][][t]{0.57\linewidth}
        \includegraphics[width=\textwidth]{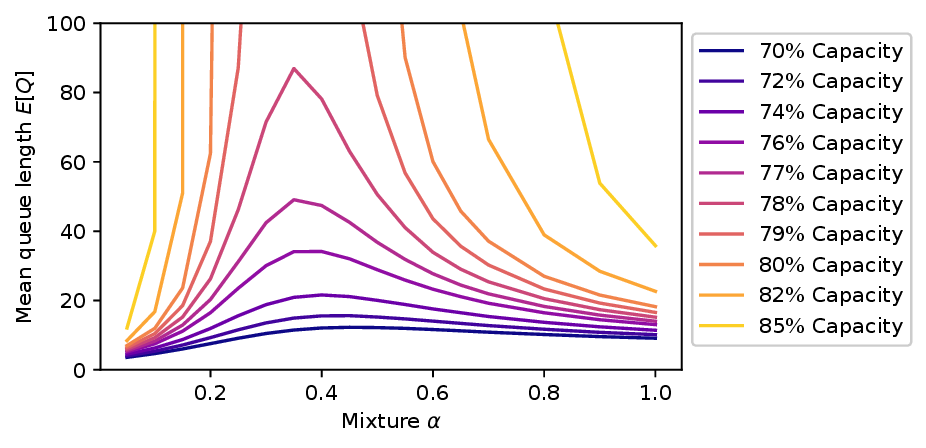}
        \caption{Simulation results ($10^8$ jobs per data point) for mean queue length when $\lambda$ is chosen so that a  constant fraction of capacity is utilized and $n=10$ servers.}
        \label{fig:half-sim}
    \end{subfigure}
    \caption{Calculations and simulation for the half-size large jobs setting with server needs $1$ and $n/2$, $\mu_1 = \mu_{n/2} = 1$ and  $p_n = 1/n^\alpha$.}
    \label{fig:half-calc-sim}
\end{figure}

\textcolor{black}{In this section, we investigate the case where the large-resource jobs require $n/2$ servers, rather than $n$ servers as in our asymptotic results.
This mimics the definition of leadership-class jobs on the Frontier supercomputer at Oak Ridge National Labs, which consists of jobs that require at least 60\% of the capacity of the supercomputer, giving top priority to such large-resource jobs and setting policy goals to attain a large fraction of load devoted to leadership-class jobs \cite{Frontier}.}

\textcolor{black}{With half-size large jobs, we can no longer use the exact results in \cref{lem:explicit-mu,lem:explicit-delta} to compute throughput and scaled queue length. Instead, we compute the throughput $\mu$ using the product-form stationary distribution formula for the MSJ 2-class saturated system from \citet[Section 5]{grosof_new_2023}, \cite[Section 5.1]{comte_graph_2025}. \cref{fig:half-calc-mu} shows the result of this calculation, with throughput increasing monotonically with increasing $\alpha$, as the mixture of load moves from $n$-server dominated to 1-server dominated, with an inflection point around $\alpha=1$.}

\textcolor{black}{Since we do not have an exact formula for limiting scaled mean queue length $\E[\Delta(Y_d)]$, we instead turn to simulation in \cref{fig:half-sim}, again with $n=10$ servers. Here we see the same bell-curve behavior of mean queue length $\E[Q]$ near the boundary of the stability region, with a peak around $\alpha=0.4$. Note that in this setting, the fraction of capacity in use is significantly higher than in any of the other setting around the boundary of the stability region for $\alpha$ near that peak -- around $78\%$ of capacity in this setting, compared to $50-60\%$ of capacity in the other settings. This follows from the fact that this system will never leave $n/2$ or more servers empty, while settings with jobs that demand $n$ servers may leave up to $n-1$ servers empty. Note also that in this setting, the peak of mean queue length $\E[Q]$ does not occur at the same $\alpha$ value as the inflection point of throughput, in contrast to \cref{sec:emp-duration-scaling,sec:emp-original}. We leave the analysis of this behavior to future work.}

\subsection{Calculations and Simulation for Three-class System}
\label{sec:emp-three}

\begin{figure}
    \centering
        \includegraphics[width=0.57\textwidth]{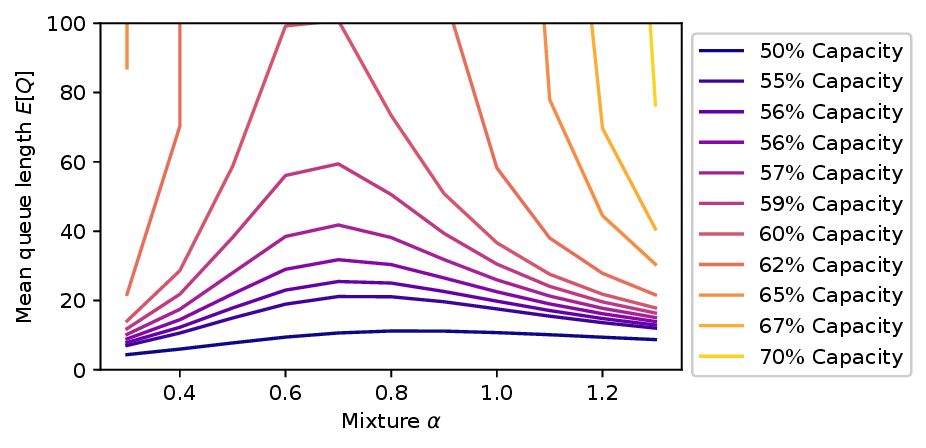}
        \caption{Simulation results ($10^8$ jobs per data point) for three-class setting with server needs $1, 5$, $10$, $n=10$ servers, $\mu_1 = \mu_5= \mu_{10} = 1$, and  $p_5 = p_{10} = n^\alpha/2$.}
    \label{fig:three-sim}
\end{figure}

\textcolor{black}{In this section, we investigate the case where there are three classes of jobs, requiring $1, n/2,$ and $n$ servers,  respectively. This is the most realistic setting because in practice, queueing systems where jobs require varying resources are not limited to two possibilities, but instead vary widely.}

\textcolor{black}{In this setting, it is more difficult to calculate throughput $\mu$ than in the half-size large jobs setting of \cref{sec:emp-half}. There is an exact product-formula stationary distribution formula for this setting, given by \citet{rumyantsev_stability_2017}, as well as \cite[Section 3]{grosof_new_2023}. However, the number of terms in this formula is exponential in $n$, while the formula in \cref{lem:explicit-mu} and the formula used in \cref{sec:emp-half} each contained a number of terms that was linear in $n$. We leave the problem of efficiently calculating $\mu$ in this setting to future work -- a formula with a number of terms that is linear in $n$ is likely doable. No exact formula for $\E[\Delta(Y_d)]$ is known in this setting.}

\textcolor{black}{We now turn to simulation in \cref{fig:three-sim}, again with $n=10$ servers. There is a larger parameter space of probability distributions over the three classes, but we choose arbitrarily to equalize the fraction of jobs in each of the two larger-resource classes. Here we see the same bell-curve behavior of mean queue length $\E[Q]$ near the boundary of the stability region, with a peak around $\alpha=0.7$, between the peaks in the original setting in \cref{sec:emp-original} and the half-size large jobs setting in \cref{sec:emp-half}. The fraction of capacity in use is likewise between those two settings around the boundary of the stability region near that peak. Even for this more complicated setting, our key bell-curve insight is again confirmed.}

\section{Conclusion} \label{conclusions}

In this paper, we studied the multiserver-job (MSJ) setting, and specifically the 1-and-$n$ system, in the load-focused multilevel scaling limit.
Within this limit, we provided the first analysis of the asymptotic growth rate of throughput $\mu$ and scaled mean queue length $\E[Q(1-\rho)] \to \E[\Delta(Y_d)]$
in each of three regimes: the $n$-server dominated load regime, balanced load regime, and 1-server dominated load regimes.
\textcolor{black}{Our \cref{thm:n-server-dominated,thm:1-server-dominated}
demonstrate} that the balanced load regime represents a tipping point with respect to throughput: In the $n$-server dominated regime, throughput is dominated by 1-server jobs that run while an $n$-server job is blocking the head of the queue, while in the 1-server dominated regime, throughput is dominated by $1$-server jobs running on all $n$ servers.
Similarly, we prove that the balanced load regime exhibits the largest scaled response time, indicating the most variation in service rates: $\E[\Delta(Y_d)] = \theta(n)$ in this regime, with slower growth rates as a function of $n$ in the other two regimes.

\textcolor{black}{Furthermore, our numerical and simulation results indicate that these insights can be generalized to duration-scaled systems where low-resource jobs require much lower service duration than high-resource jobs, systems where the large-resource jobs require $n/2$ servers rather than $n$ servers, and systems with three classes where jobs require $1$, $n/2$, and $n$ servers. A natural direction for future research would be to prove asymptotic results in these settings.}


\section{Competing interests declaration}
Competing interests: The authors declare none.

\section*{Acknowledgement}
The work of the first author was supported by a Tennenbaum Postdoctoral Fellowship at Georgia Tech and by AFOSR Grant FA9550-24-1-0002. The work of the second author was supported by the National Science Foundation Grant CMMI-2127778.

\bibliographystyle{plainnat}
\bibliography{refs}

\appendix
\section{Additional numerical comparisons}
\label{app:additional-numerical}

In this section, we numerically confirm our theoretical results \textcolor{black}{in \cref{thm:n-server-dominated,thm:1-server-dominated}, showing} the convergence of $\mu$ and $\E[\Delta(Y_d)]$ to our theoretically predicted asymptotic behaviors. In \cref{sec:empirical}, we demonstrated this convergence for the parametrization $\mu_1 = \mu_n = 1$. In this section, we explore the cases $\mu_1 = 10, \mu_n = 1$, where 1-server jobs have much shorter durations than $n$-server jobs, and $\mu_1=1, \mu_n= 10$, where 1-server jobs have much longer durations than $n$-server jobs.

In \cref{fig:n-dominated-10-1,fig:balanced-10-1,fig:1-dominated-10-1}, we explore the case where $\mu_1 = 10, \mu_n = 1$, under the parametrizations $p_n = n^{-0.5}, p_n= n^{-1},$ and $p_n = n^{-2}$, respectively.

\begin{figure}
    \centering
    \begin{subfigure}{0.49\linewidth}
        \includegraphics[width=\textwidth]{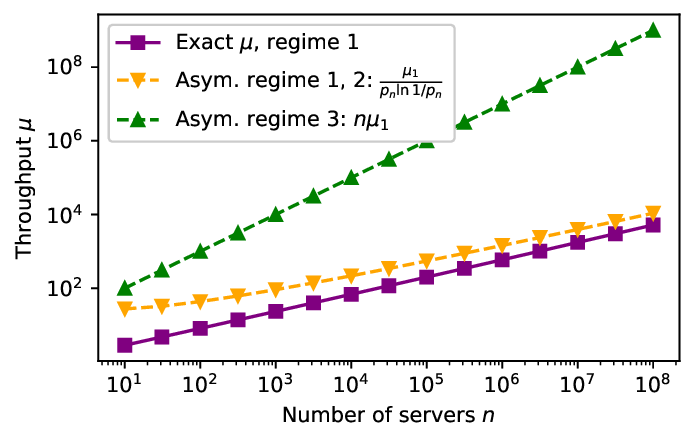}
    \end{subfigure}
    \begin{subfigure}{0.49\linewidth}
        \includegraphics[width=\textwidth]{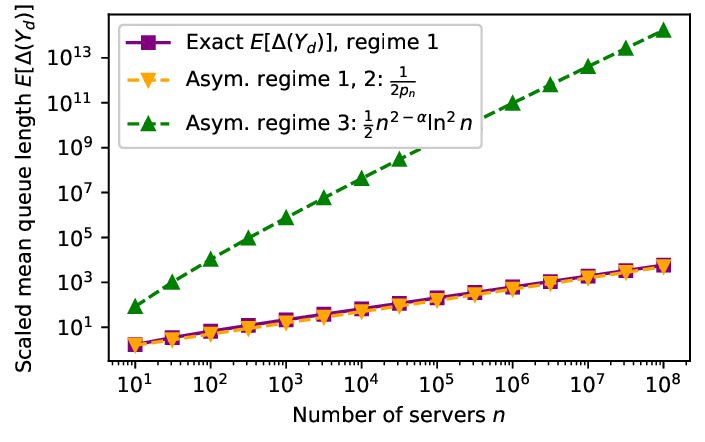}
    \end{subfigure}
    \caption{Exact versus asymptotic formulas in the $n$-server dominated regime for $\mu$ and  $\E[\Delta(Y_d)]$ as functions of  $n$. Parametrization: $p_n=n^{-0.5}$ and $\mu_1=10, \mu_n=1$.}
    \label{fig:n-dominated-10-1}
\end{figure}

\begin{figure}
    \centering
    \begin{subfigure}{0.49\linewidth}
        \includegraphics[width=\textwidth]{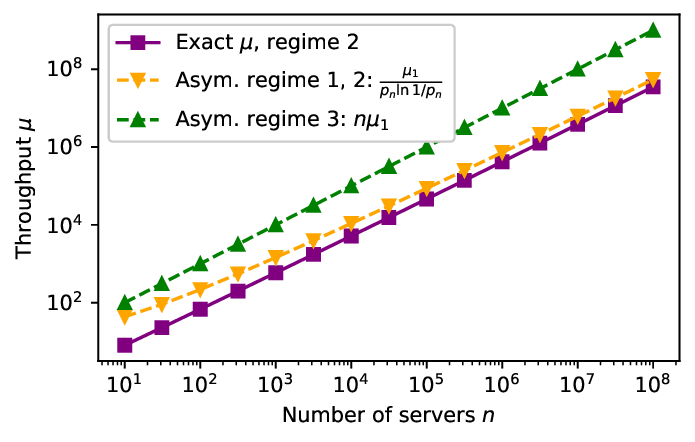}
    \end{subfigure}
    \begin{subfigure}{0.49\linewidth}
        \includegraphics[width=\textwidth]{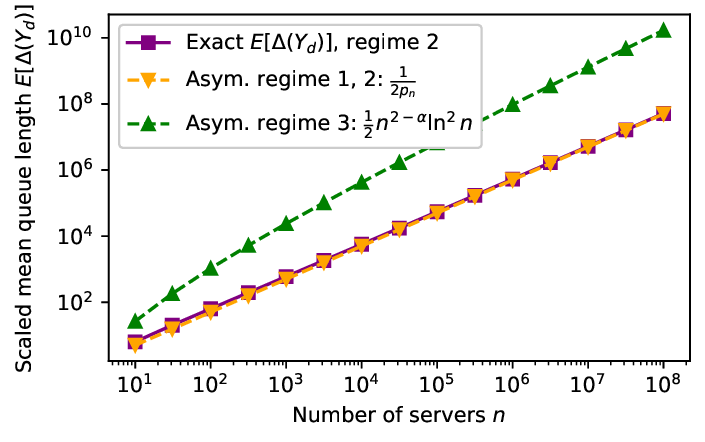}
    \end{subfigure}
    \caption{Exact versus asymptotic formulas in the balanced load  regime for $\mu$ and  $\E[\Delta(Y_d)]$ as functions of  $n$.
    Parametrization: $p_n=n^{-1}$ and $\mu_1=10, \mu_n=1$.}
    \label{fig:balanced-10-1}
\end{figure}

\begin{figure}
    \centering
    \begin{subfigure}{0.49\linewidth}
        \includegraphics[width=\textwidth]{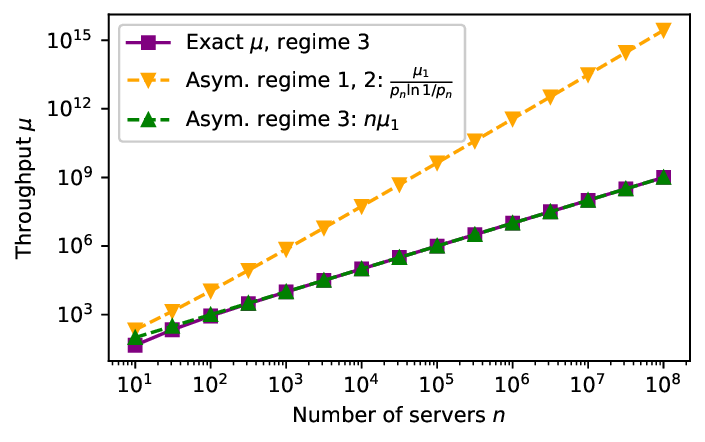}
    \end{subfigure}
    \begin{subfigure}{0.49\linewidth}
        \includegraphics[width=\textwidth]{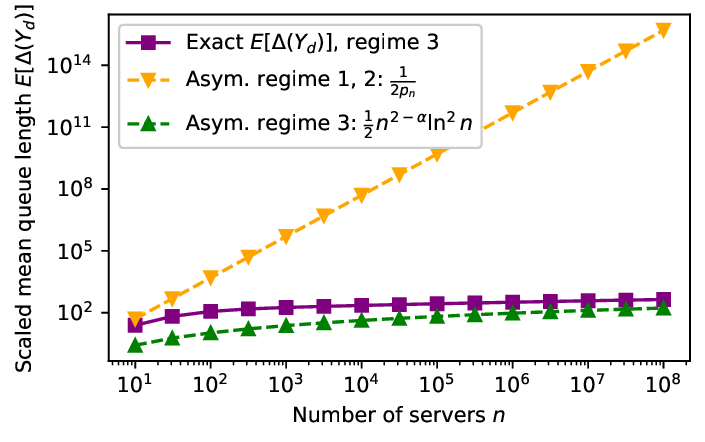}
    \end{subfigure}
    \caption{Exact versus asymptotic formulas in the $1$-server dominated regime for $\mu$ and  $\E[\Delta(Y_d)]$ as functions of  $n$.
    Parametrization: $p_n=n^{-2}$ and $\mu_1=10, \mu_n=1$.}
    \label{fig:1-dominated-10-1}
\end{figure}

These figures demonstrate the same multiplicative convergence shown in \cref{sec:empirical} and \textcolor{black}{proven in \cref{thm:n-server-dominated,thm:1-server-dominated}:} In \cref{fig:n-dominated-10-1,fig:balanced-10-1}, the orange and purple curves converge, while in \cref{fig:1-dominated-10-1}, the green and orange curves converge.

Throughout the paper, for the benefit of monochromatic readers, the green curves use upward-triangle markers, the orange curves use downward-triangle markers, and the purple curves use square markers.

In some cases, the convergence is faster or slower than in the $\mu_1 = \mu_n = 1$ setting shown in \cref{sec:empirical}. For instance, in \cref{fig:n-dominated-10-1}, the convergence for $\mu$ to its asymptotic curve is slower than it was in \cref{fig:n-dominated}, while the convergence for $\E[\Delta(Y_d)]$ is faster than in \cref{fig:n-dominated}, and in fact the asymptotic expression is a tight approximation across all $n$.

As for why the quality of approximation changes, for the $\mu$ plots, the two asymptotic formulas emphasize terms corresponding to periods of time when 1-server jobs are in service. The orange asymptotic curve emphasizes the period when 1-server jobs are emptying from service to allow an $n$-server job to be served, while the green asymptotic curve emphasizes the period when 1-server jobs are being served on all $n$ servers.
In the $\mu_1 = 10, \mu_n = 1$ parameterization, the 1-server jobs have shorter durations, raising the importance of the time when the $n$-server jobs are in service, which is not asymptotically relevant, but shows up in these plots for small $n$.

For the $\E[\Delta(Y_d)]$ plots, we do not have a clear interpretation -- in \cref{fig:n-dominated-10-1,fig:balanced-10-1}, the orange asymptotics are tighter than in \cref{sec:empirical}, while in \cref{fig:1-dominated-10-1}, the green asymptotics are looser.

\begin{figure}
    \centering
    \begin{subfigure}{0.49\linewidth}
        \includegraphics[width=\textwidth]{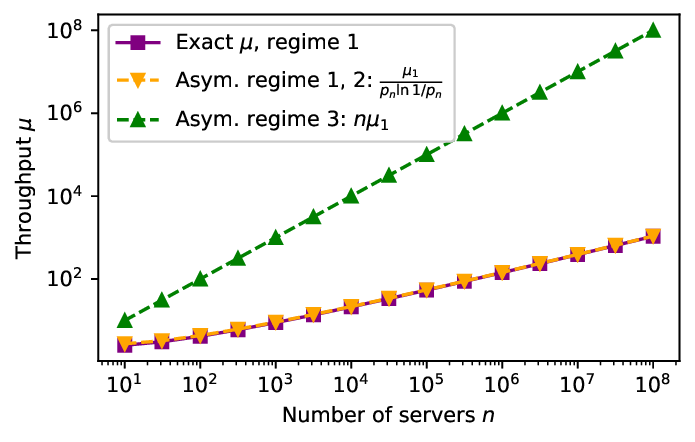}
    \end{subfigure}
    \begin{subfigure}{0.49\linewidth}
        \includegraphics[width=\textwidth]{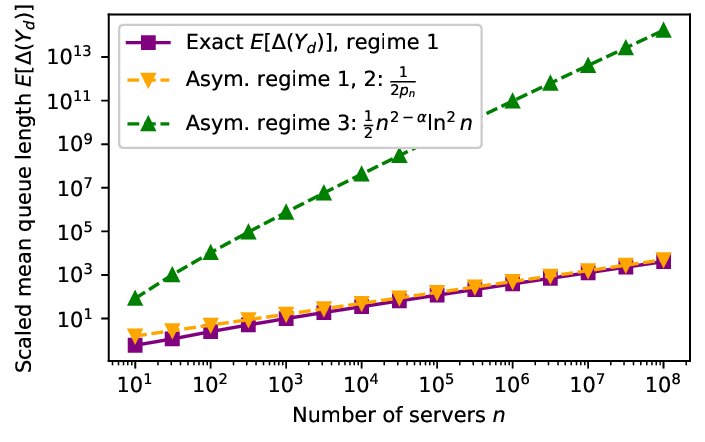}
    \end{subfigure}
    \caption{Exact versus asymptotic formulas in the $n$-server dominated regime for $\mu$ and  $\E[\Delta(Y_d)]$ as functions of  $n$.
    Parametrization: $p_n=n^{-0.5}$ and $\mu_1=1, \mu_n=10$.}
    \label{fig:n-dominated-1-10}
\end{figure}

\begin{figure}
    \centering
    \begin{subfigure}{0.49\linewidth}
        \includegraphics[width=\textwidth]{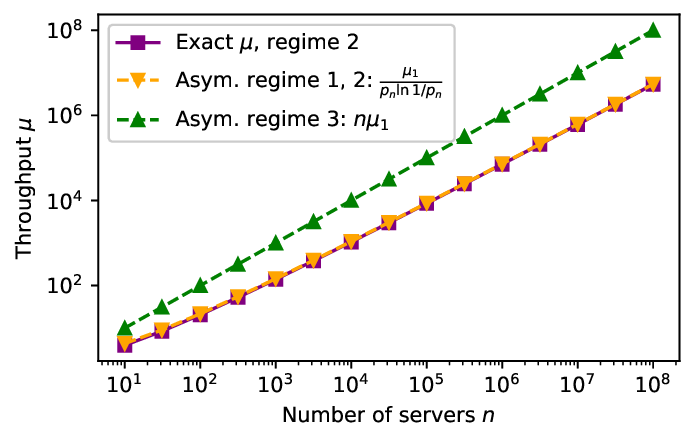}
    \end{subfigure}
    \begin{subfigure}{0.49\linewidth}
        \includegraphics[width=\textwidth]{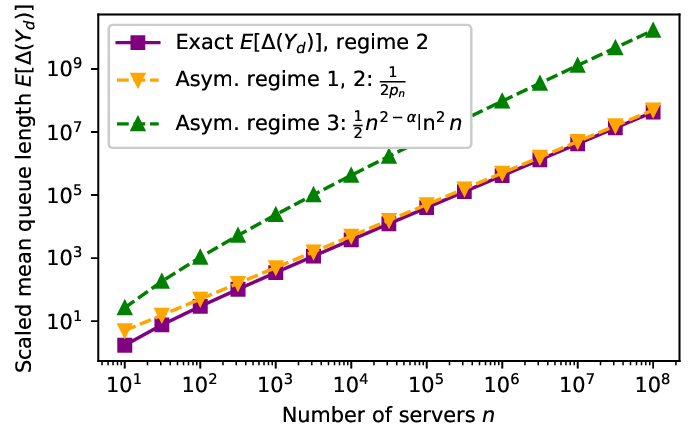}
    \end{subfigure}
    \caption{Exact versus asymptotic formulas in the balanced load  regime for $\mu$ and  $\E[\Delta(Y_d)]$ as functions of  $n$.
    Parametrization: $p_n=n^{-1}$ and $\mu_1=1, \mu_n=10$.}
    \label{fig:balanced-1-10}
\end{figure}

\begin{figure}
    \centering
    \begin{subfigure}{0.49\linewidth}
        \includegraphics[width=\textwidth]{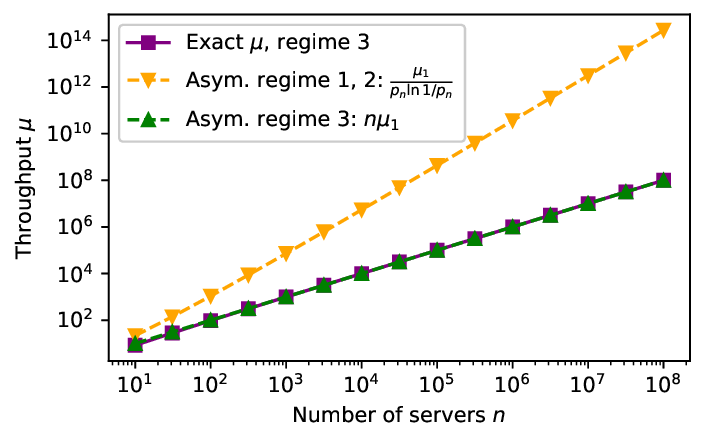}
    \end{subfigure}
    \begin{subfigure}{0.49\linewidth}
        \includegraphics[width=\textwidth]{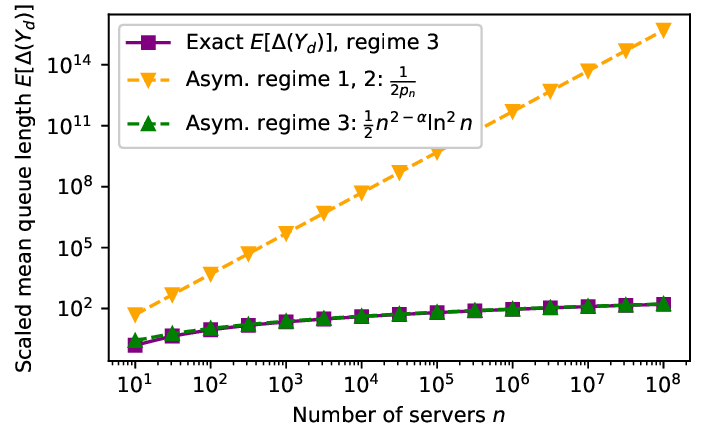}
    \end{subfigure}
    \caption{Exact versus asymptotic formulas in the $1$-server dominated regime for $\mu$ and  $\E[\Delta(Y_d)]$ as functions of  $n$.
    Parametrization: $p_n=n^{-2}$ and $\mu_1=1, \mu_n=10$.}
    \label{fig:1-dominated-1-10}
\end{figure}

Likewise, in figures \cref{fig:n-dominated-1-10,fig:balanced-1-10,fig:1-dominated-1-10}, we explore the case where $\mu_1 = 1, \mu_n = 10$, under the parametrizations $p_n = n^{-0.5}, p_n= n^{-1},$ and $p_n = n^{-2}$, respectively.

As before, in \cref{fig:n-dominated-1-10,fig:balanced-1-10}, the orange and purple curves converge, while in \cref{fig:1-dominated-1-10}, the green and orange curves converge.

As predicted by the discussion above, the exact formulas for $\mu$ are very tight in all three regimes, as this parametrization lengthens the periods when 1-server jobs are being served, improving the accuracy of the asymptotics at low $n$. 

Again, it is unclear when our asymptotic formulas for $\E[\Delta(Y_d)]$ are tight for lower values of $n$. In this case, the green and purple curves are very close in \cref{fig:1-dominated-1-10}, the reverse of the situation in \cref{fig:1-dominated-10-1}.
\end{document}